\documentclass[a4paper]{article}

\usepackage{tikz}
\usetikzlibrary {calc}
\usetikzlibrary {arrows.meta} 
\usetikzlibrary {shapes.geometric}
\usetikzlibrary {graphs,shapes.geometric} 
\usetikzlibrary {decorations.pathmorphing}

\usetikzlibrary {graphs}

\usepackage[shortlabels]{enumitem}
\usepackage[]{geometry}
\usepackage[utf8]{inputenc}
\usepackage{fancybox, calc}
\usepackage{authblk}
\usepackage{amssymb}
\usepackage{amsmath}
\usepackage{hyperref}
 \usepackage{systeme}
\usepackage{graphicx}
\usepackage{enumitem}
\usepackage{ucs}
\usepackage{amsfonts}
\usepackage{amssymb}
\usepackage{makeidx}
\usepackage{diagbox}
\usepackage{tabularx}
\usepackage{float}
\usepackage{url}
\usepackage{ marvosym }
\usepackage{ wasysym }
\usepackage{ mathrsfs }
\usepackage[utf8]{inputenc}
\usepackage[english]{babel}
\usepackage[normalem]{ulem}
\usepackage{ucs}
\usepackage{amsfonts}
\usepackage{dcolumn}
\usepackage{makeidx}
\usepackage{bm}
\usepackage{multirow}
\usepackage{graphicx}
\usepackage{hyperref}
\usepackage{tikz}
\usepackage{subfig}
\usepackage{booktabs}
\usepackage{xcolor}
\usepackage{fancyhdr}
\usepackage{amsmath, amsthm, amssymb}
\usepackage{multicol}
\usepackage{yfonts}
\usepackage{multicol}
\usepackage{caption}
\usepackage{cancel}
\usepackage{amssymb,amsfonts,amssymb,bbm}
\usepackage{amsmath}
\usepackage{phaistos,hieroglf,staves,ifsym,marvosym,skull}
\usepackage{tikz}
\usepackage{transparent}
\usepackage{eso-pic}
\usepackage{soul}
\usepackage{cite}
\usepackage{lettrine}
\date{\today}

\def\Rset{\mathbb{R}}
\def\pdf{f}
\def\supp{\Omega}
\newcommand{\down}[1]{\mathfrak {D}_{#1}}
\newcommand{\adown}{{\pdf}^\downarrow_{\alpha}}
\newcommand{\up}[1]{\mathfrak {U}_{#1}}
\newcommand{\aup}{{\pdf}^\uparrow_{\alpha}}

\def\Rset{\mathbb{R}}

\def\esssup{\operatorname{esssup}}

\def\sign{\operatorname{sign}}

\def\gauss{g}

\def\pdf{f}

\newcommand{\descort}[2][]{\mathfrak{E}_{#1}\ifthenelse{\isempty{#2}}{}{ {\left[ #2 \right]}}}
\newcommand{\escort}[2][]{\varepsilon_{#1}\ifthenelse{\isempty{#2}}{}{ {\left[ #2 \right]}}}

\newtheorem{theorem}{Theorem}[section]

\newtheorem{lemma}{Lemma}[section]
\newtheorem{definition}{Definition}[section]
\newtheorem{proposition}{Proposition}[section]

\newtheorem{remark}{Remark}[section]

\numberwithin{equation}{section}

\definecolor{rougeG}{rgb}{.76,0,.12}
\definecolor{vertG}{rgb}{.07,.56,.25}

\newcommand{\DP}[1]{{\color{vertG} #1}}

\title{Generalized informational functionals and new monotone measures of statistical complexity}
\author[1]{Razvan Gabriel Iagar}
\affil[1]{Departamento de Matemática Aplicada, Ciencia e Ingeniería de los Materiales y Tecnología Electrónica, Universidad Rey Juan Carlos,
		28933 Móstoles (Madrid), Spain}
\author[1,2]{David Puertas-Centeno}

\affil[2]{Data, Complex Networks and Cybersecurity Research Institute, Universidad Rey Juan Carlos, 28028 (Madrid), Spain}

\usepackage{yfonts}

\usepackage{blindtext}

\usepackage[T1]{fontenc}
\usepackage{palatino}
\begin{document}
	\maketitle
	\begin{abstract}
In this paper we introduce a biparametric family of transformations which can be seen as an extension of the so-called up and down transformations. This new class of transformations allows to us to introduce new informational functionals, which we have called \textit{down-moments} and \textit{cumulative upper-moments}. A remarkable fact is that the down-moments provide, in some cases, an interpolation between the $p$-th moments and the power Rényi entropies of a probability density. We establish new and sharp inequalities relating these new functionals to the classical informational measures such as moments, Rényi and Shannon entropies and Fisher information measures. We also give the optimal bounds as well as the minimizing densities, which are in some cases expressed in terms of the generalized trigonometric functions. We furthermore define new classes of measures of statistical complexity obtained as quotients of the new functionals, and establish monotonicity properties for them through an algebraic conjugation of up and down transformations. All of these properties highlight an intricate structure of functional inequalities.
\end{abstract}

\bigskip


\section{Introduction}

The study of the mathematical properties of informational functionals has been a hot research topic during the last decades (see for example~\cite{Jizba2004,Jizba2004b,Masi2005,Jizba2019,Corominas2024}, surveys such as~\cite{Verdu1998,Amigo2018} or the classical monograph \cite{Gray2011}). This long term development provides, on the one hand, theoretical results, among which establishing sharp inequalities between functionals and measures is a distinguished research line \cite{Dembo1991, Lutwak2005, Bercher2012}. On the other hand, interesting tools for a wide class of applications in physics, computer science and engineering have been obtained along the years, in particular in communication theory~\cite{Shannon1948}, signal proccesing~\cite{Cover2006} or quantum mechanics~\cite{Bialynicki1975}, among a large number of other areas of science and technology where information theory plays a significant role.

Beyond the standard moments, the Shannon and Rényi entropies represent the most studied measures in information theory and their understanding became a basis for further theoretical constructions. Thus, a huge variety of other informational functionals have been defined in connection to applications in different contexts, such as the Tsallis entropy in non-extensive thermodynamics~\cite{Tsallis1988} or the very fashionable nowadays Kaniadakis entropy, which allows to generalize the standard Boltzmann-Gibbs statistical mechanics within the framework of the special relativity~\cite{Kaniadakis2002,Kaniadakis2005}. It is worth to mention that, as stated in Enciso and Tempesta work~\cite{Enciso2017}, the only composable generalized entropy in trace form is the Tsallis entropy. Going one step forward, the notion of entropy has been extended to far more general frameworks and thus new functionals have been introduced. Let us only mention here the so-called $\phi$-entropies~\cite{Toranzo2018} constructed by generalizing the definition of the Shannon entropy by employing more general functions with good mathematical properties, and the group entropies ~\cite{Tempesta2011,Curado2016}. These latter entropies have been applied to systems whose phase space grows (in asymptotic sense) sub-exponentially, exponentially or super-exponentially with respect to their numbers of elements~\cite{Jensen-Tempesta2018}.

A different research direction has been to construct measures of statistical complexity. The aim of such measures is to grasp essential properties of the probability distribution of a system. The discussion about the minimal mathematical properties that a complexity measure must satisfy is yet an open problem. However, an interesting perspective is raised through the notion of monotonicity of a complexity measure~\cite{Rudnicki2016}. One of the most successful complexity measures in the literature is the one introduced by López, Mancini and Calbet (called LMC complexity measure) \cite{LopezRuiz1995, LopezRuiz2005}, generalized later in~\cite{Sanchez-Moreno2014}. The monotonicity property of this family of measures, stated initially as an open problem in~\cite{Rudnicki2016}, is proved by one of the authors through the introduction of the notion of \textit{differential-escort} transformations~\cite{Puertas2019}. Besides this application, these transformations have been the key points for establishing triparametric generalizations of the Stam, Cramér-Rao and moment-entropy inequalities~\cite{Puertas2019,Puertas2025}.

\medskip

Very recently, the authors introduced a new pair of mutually inverse transformations between probability densities, called up and down transformations~\cite{IP2025}. It has been proved in the mentioned reference that the $p$-th moment of a down transformed density corresponds to the power Rényi entropy of the original probability density. Moreover, the power Rényi entropy of the down transformed density corresponds to a generalized Fisher information measure defined by Lutwak~\cite{Lutwak2005}. The application of these transformations to the well established informational inequalities~\cite{Lutwak2005,Bercher2012} has allowed to reveal a \textit{mirrored domain} of application of the latter inequalities. By mirrored domain, we understand a domain of parameters disjoint from the classical one, in which the corresponding inequality, with the same mathematical expression, holds true, but applying to a different class of densities and usually with a different optimal constant.

An interesting fact put into evidence by the up and down transformations is that, in this mirrored domain, the minimizing densities of the newly derived sharp inequalities exhibit a divergent behavior at the border of their support, which is in a stark contrast with respect to the regular behavior of the minimizing densities in the \textit{classical domain} of the generalized moment-entropy, Stam and Cramér-Rao inequalities. Moreover, a large family of the new minimizing densities has been explicitly expressed in terms of the generalized biparametric trigonometric functions defined by Drábek and Manásevich~\cite{Drabek1999}. Applications of this new class of transformations to the entropy-like and Fisher-like Hausdorff moment problems have been also found.

In a second work~\cite{IP2025(b)}, the authors explore further applications of the above mentioned transformations by introducing new classes of informational functionals, called upper-moments and down-Fisher measures according to the transformation employed in their derivation. While the upper-moments involve the expected value of a weighted cumulative function, the structure of the down-Fisher measures involves not only the first derivative of the probability density (as the usual Fisher information measure and its generalizations do, see~\cite{Lutwak2005}), but also its second derivative. The application of the up and down transformations to previous informational inequalities allows to derive new sharp inequalities for the latter functionals and deduce their optimal bounds and minimizers by expressing them as functions of the optimal bounds of the classical inequalities; remarkably, the minimizers of these new inequalities include the family of Beta distributions. Furthermore, a very  surprising fact is that we have found sharp upper bounds for the classical inequalities, depending on the regularity conditions of the probability densities involved in the inequalities. In particular, the results in \cite{IP2025(b)} taken as a whole allowed us to discover an interesting  structure composed by different levels of informational functionals and inequalities and the up and down transformations are the tools for moving from one level to another one inside this structure.

\medskip

In this work we continue with this research line and introduce a biparametric version of the up and down transformations, which turns out to be the composition of the differential-escort and the original (one-parametric) up and down transformations respectively. Indeed, the biparametric down transformation contains both the down and the differential-escort transformations as particular cases. This allows us to understand the second parameter as an \emph{interpolation parameter}, performing a continuous ``deformation" from the differential-escort to the down transformation.

This interpolation refines and strongly enriches the previously established structure. For example, a remarkable fact is that our interpolation between transformations achieved through the second parameter allows to connect the classical and mirrored domains of the parameters, which, as explained above, have been observed in a separate manner. This connection is seen at the level of the informational inequalities introduced in Section \ref{sec:inequalities}. As we shall see, the classical and mirrored domains of parameters of different inequalities become particular cases of a single, unified inequality interpolating between them. Nevertheless, we stress that the classical and mirrored domain of a same inequality still remain disjoint. Let us mention here that, in mathematical analysis, interpolation techniques represented a breakthrough in areas such as harmonic analysis, measure theory, functional inequalities and operator theory (see for example the mathematical monograph \cite{Bennett1988}). We hope and expect that our interpolation might find interesting applications in information theory.

We apply the biparametric family of up and down transformations to introduce two new families of informational functionals and to prove a number of sharp inequalities involving them, together with their optimal constants and minimizing densities. More precisely, the first family of new functionals, which we have called \emph{cumulative upper-moments}, consists in the expected value of a doubly weighted cumulative function, involving a combination of three parameters. The second class of informational functionals, which we have called \emph{down-moments}, consists in the expected value of a kind of cumulative function including the derivative of the density. A very remarkable fact related to the down-moments is that they are functionals which, under some conditions on the densities, provide an intricate continuous interpolation between the Rényi entropy power and the absolute $p$-th moments. As a byproduct of our biparametric transformations, we also prove a mirrored version of the cumulative moment-entropy inequality, completing thus the collection of mirrored versions of classical inequalities started in \cite{IP2025}.

We also employ the latter new functionals, together with the generalized Fisher information measures and the generalized $p$-th moments, in order to define new generalized measures of statistical complexity. Moreover, we show that different compositions of the above mentioned transformations allow us to prove the monotonicity properties of these new complexity measures in the sense of Rudnicki et al~\cite{Rudnicki2016}. In particular, the latter combinations of transformations are in fact algebraic conjugations of the differential-escort transformations with mutually inverse up and down transformations, and thus they inherit the group structure stemming from the properties of the differential-escort transformations.

\medskip

The structure of this work is the following: after a preliminary Section~\ref{sec:preliminary} where we gather previously introduced notions and results employed throughout the paper for the reader's convenience, the bi-parametric up and down transformations, together with the notions of cumulative upper-moments and generalized down-moments are defined in Section~\ref{sec:transf}. Their basic mathematical properties, emphasizing on the order relations with respect to some of their parameters, are also given in Section~\ref{sec:transf}. The paper continues with Section~\ref{sec:inequalities}, dedicated to the new informational inequalities satisfied by the latter informational measures. In Section~\ref{sec:monot} we introduce new complexity measures motivated by the new informational functionals and we prove a monotonicity property for each of them. The paper is completed by a section focusing on the conclusions of our analysis and by an Appendix where compositions of the above mentioned transformations are carefully calculated.

\section{Preliminary tools}\label{sec:preliminary}

This section is dedicated to a brief recall of the concepts such as informational functionals and measures, inequalities between them, generalized trigonometric functions, previously defined transformations, and some basic properties of them, all of them employed throughout the paper.

\subsection{Basic informational measures}

We first recall a standard notation that will be employed frequently in this work. Given a probability density function $f$, we denote throughout the paper the expected value
$$
\left\langle A\right\rangle_f:=\int_{\Rset}f(x)A(x)\,dx,
$$
for any function $A$ for which the above integral is finite. By convention and for simplicity, the integrals will be generally written over $\Rset$ (unless otherwise indicated), understanding that, if a density function is compactly supported, the integral is taken on its support.

\medskip

\noindent \textbf{The $p$-th absolute moment} of a probability density function $\pdf$, with $p \geqslant 0$, is defined as
\begin{equation}\label{eq:pabs}
\mu_p[\pdf] = \int_\Rset |x|^p \, \pdf(x) \, dx  = \left\langle  \, |x|^p \, \right\rangle_\pdf.
\end{equation}
If the $p$-th absolute moment is finite, we also introduce the $p$-th deviation
\begin{equation*}
\begin{split}
&\sigma_p[\pdf] =\left(\int_\Rset |x|^p \, \pdf(x) \, dx \right)^{\frac{1}{p}}, \quad {\rm for} \ p > 0, \\
&\sigma_0[\pdf] = \lim\limits_{p \to 0} \sigma_p[\pdf]  = \exp\left(\int_\Rset \pdf(x) \, \log|x| \, dx \right), \\
&\sigma_{\infty}[\pdf]=\lim\limits_{p\to\infty}\sigma_p[\pdf]=\esssup\big\{|x| : x\in\Rset, \pdf(x) > 0 \big\}.
\end{split}
\end{equation*}
Although in the classical theory $\mu_p$ and $\sigma_p$ are not defined for exponents $p<0$, we extend the definition to negative moments for those density functions for which these moments are finite.

\medskip

\noindent \textbf{R\'enyi and Tsallis entropies.} Given $\lambda\geqslant1$, the differential R\'enyi and Tsallis entropies of $\lambda$-order of a probability density function $\pdf$ are defined as
\begin{equation*}
R_\lambda[\pdf] = \frac1{1-\lambda} \log\left( \int_\Rset [\pdf(x)]^\lambda \, dx \right), \qquad T_\lambda[\pdf] = \frac1{\lambda-1} \left( 1 - \int_\Rset [\pdf(x)]^\lambda \, dx \right),
\end{equation*}
as introduced in \cite{Renyi1961, Tsallis1988}. In the limiting case $\lambda=1$ we recover the well-known differential Shannon entropy
\begin{equation*}
\lim\limits_{\lambda \to 1}R_{\lambda}[\pdf]=\lim\limits_{\lambda \to 1} T_{\lambda}[\pdf]=S[\pdf]=-\int_\Rset \pdf(x) \, \log\pdf(x) \, dx.
\end{equation*}
For easiness, throughout the paper we employ the R\'enyi entropy power, which is the exponential of the R\'enyi entropy,
$$
N_\lambda[\pdf] = e^{R_{\lambda}[\pdf]} = \left\langle\pdf^{\lambda-1}(x)\right\rangle^\frac1{1-\lambda}_\pdf.
$$
Similarly, $N[\pdf]=e^{S[\pdf]}$ designs the Shannon entropy power. It is an easy fact that the R\'enyi and Tsallis entropies are one-to-one mapped by
$$
T_{\lambda}[\pdf]=\frac{e^{(1-\lambda) R_{\lambda}[\pdf]} - 1}{1 - \lambda},
$$
and thus any result involving the R\'enyi entropy can be readily expressed in terms of the Tsallis entropy as well. This is why, we will only work with the R\'enyi entropy power throughout the paper.

\medskip

\noindent \textbf{$(p,\lambda)$-Fisher information.} This is an informational functional acting on derivable probability density functions and it was introduced by Lutwak and Bercher~\cite{Lutwak2005, Bercher2012, Bercher2012a} as an extension to more general exponents of the usual Fisher information. Given $p>1$ and $\lambda \in \Rset^*$, the $(p,\lambda)$-Fisher information of a probability density function $f$ is defined as
\begin{equation}\label{eq:def_FI}
F_{p,\lambda}[\pdf]=\int_\Rset \left|\pdf^{\lambda-2}(x) \, \frac{d\pdf}{dx}(x)\right|^{p} \pdf(x) \, dx
\end{equation}
whenever $\pdf$ is differentiable on the closure of its support. In particular, the standard Fisher information is recovered as the $(2,1)$-Fisher information. We shall sometimes employ the related functional
\begin{equation}\label{eq:def_FI_bis}
\phi_{p,\lambda}[\pdf]=\big(F_{p,\lambda}[\pdf]\big)^{\frac{1}{p\lambda}}.
\end{equation}
The Fisher information is defined in \cite{Lutwak2005, Bercher2012, Bercher2012a} only for exponents $p>1$. However, the authors have extended in recent works \cite{IP2025, IP2025(b)} the definitions \eqref{eq:def_FI} and \eqref{eq:def_FI_bis} to exponents $p<1$ (and even $p<0$), for those classes of density functions for which the $(p,\lambda)$-Fisher information is finite. In particular, in the previously mentioned works, we have considered also probability density functions defined on a bounded interval and divergent at its limits, for which precisely negative values of the exponent $p$ are more expected to keep finite the expression of $F_{p,\lambda}[f]$ in Eq. \eqref{eq:def_FI}.

\subsection{Generalized trigonometric functions}\label{sec:GTFs}

We give in this short section a brief recap of the generalized trigonometric functions introduced by Drábek and Manásevich. Such functions have been first introduced with a single parameter in \cite{Shelupsky1959} and \cite{Lindqvist1995}. In our work, we employ the $(p,q)$-generalized trigonometric functions $\sin_{p,q},\cos_{p,q}$ that, together with their respective hyperbolic counterparts, were first defined, to the best of our knowledge, in \cite{Drabek1999}. The $(p,q)$-sine function is defined as the inverse of
\begin{equation} \label{arcsin_def}
\text{arcsin}_{p,q} (z)=\int_0^z(1-t^q)^{-\frac1p}dt=z\,_2F_1\left(\frac 1p,\frac1q;1+\frac 1q;z^q\right),  \quad p,q>1,
\end{equation}
which, as we see, can also be expressed using the Gaussian hypergeometric function (see for example \cite{Edmunds2012, Bhayo2012}). The cosine function is defined as the derivative of the sine function
\begin{equation*}
\cos_{p,q}(z)=\frac{d}{dz}\sin_{p,q}(z).
\end{equation*}
From these definitions, and applying the inverse function rule, the following Pythagorean-like identity is obtained:
\begin{equation}\label{eq:pyth}
\cos_{p,q}^p(z)+\sin_{p,q}^q(z)=1.
\end{equation}
The hyperbolic counterparts of the generalized trigonometric functions are defined analogously, starting from the hyperbolic arcsine function
\begin{equation}\label{arcsinh_def}
\text{arcsinh}_{p,q} (z)=\int_0^z(1+t^q)^{-\frac1p}dt=z\,_2F_1\left(\frac 1p,\frac1q;1+\frac 1q;-z^q\right),
\end{equation}
and defining the $\sinh$ function as its inverse. Thus, we find
\begin{equation}\label{eq:pythh}
\cosh^p_{p,q}(z)-\sinh^q_{p,q}(z)=1,
\end{equation}
where
\begin{equation*}
\cosh_{p,q}(z)=(\sinh_{p,q}(z))'.
\end{equation*}
A more detailed description of the properties of the generalized trigonometric functions and more applications of them can be found in recent papers such as \cite{Gordoa2025, Puertas2025, IPT2025}.

\subsection{Differential-escort transformation and cumulative moments}\label{subsec:escort}

The differential-escort transformation is a transformation acting on the class of probability density functions, studied by one of the authors and his collaborators in \cite{Zozor2017, Puertas2019}. Since this transformation will be employed at several places in the sequel, we recall here its definition. If $f$ is a probability density function, then
\begin{equation}\label{def:escort}
\mathfrak{E}_{\alpha}[f](y)=f(x(y))^{\alpha}, \quad y'(x)=f(x)^{1-\alpha}.
\end{equation}
The interested reader can find a detailed analysis of the differential-escort transformation and further properties of it in \cite{Puertas2019, Puertas2025}. Note that this transformation is a change of variables on the probability density depending on the probability density function itself, but it is closely related to the power-type Sundman transformations, having applications in the study of differential equations. Related to the differential-escort transformation, the following cumulative moments
\begin{equation}\label{def:cum_mom}
\mu_{p,\gamma}[f]=\mu_p[\mathfrak{E}_{\gamma}[f]]=\int_{\Rset}\left|\int_{0}^x f(t)^{1-\gamma}\,dt\right|^pf(x)\,dx,
\end{equation}
have been considered as a biparametric family refining the classical $p$-th moments, and will be useful in the informational inequalities established in this paper. We also define the quantity
\begin{equation*}
\sigma_{p,\gamma}[f]:=\mu_{p,\gamma}^\frac{1}{p\gamma}[f]
\end{equation*}
and recall that
\begin{equation}\label{eq:sigmaescort}
\sigma_{p,\gamma}[\mathfrak E_\alpha[f]]=\sigma_{p,\alpha\gamma}^\alpha[ f]
\end{equation} according to~\cite{Puertas2025}. Let us mention here that we can define more general cumulative moments in the form
$$
\mu_{p,\gamma;x_0}[f]:=\int_{\Rset}\left|\int_{x_0}^x f(t)^{1-\gamma}\,dt\right|^pf(x)\,dx,
$$
for any $x_0\in\Rset$, and the inequalities established in this paper will hold true as well for such cumulative moments. However, for simplicity, we will only work with the standard cumulative moments defined in \eqref{def:cum_mom}.

We end this section with another useful property consisting in the connection between the differential-escort transformation and the Rényi entropy power, established in~\cite{Puertas2019}:
\begin{equation}\label{eq:Renyiescort}
N_\lambda[\mathfrak E_\alpha[f]]=N^\alpha_{1+(\lambda-1)\alpha}[f].
\end{equation}

\subsection{Up and down transformations, upper-moments and down-Fisher measures}\label{subsec:updown}

We next recall the recently introduced up/down transformations~\cite{IP2025}. They are a starting point and limiting case for the biparametric family of transformations introduced and investigated throughout the paper and are also essential tools in the proofs of the informational inequalities and monotonicity properties which form the main results of this work.

\medskip

\noindent \textbf{The down transformation.} The down transformation establishes a bijection between the set of decreasing density functions and a more general class of density functions, as follows:
\begin{definition}\label{def:down}
Let $\pdf: \supp\longrightarrow \Rset^+$ be a probability density function with $\supp=(x_i,x_f),$ where $-\infty<x_i<x_f\le\infty$, such that $\pdf'(x)<0,\;\forall x\in\supp$. Then, for $\alpha\in\Rset$, we define the transformation $\down{\alpha}[\pdf(x)]$ by
	\begin{equation}\label{eq:down}
	\pdf^{\downarrow}_\alpha(s)\equiv\down{\alpha}[\pdf(x)](s)=\pdf^\alpha(x(s))|\pdf'(x(s))|^{-1},\qquad s'(x)=\pdf^{1-\alpha}(x) |\pdf'(x)|.
	\end{equation}
\end{definition}
It is easy to see that $\down{\alpha}[\pdf]$ is a probability density function. Let us also observe that the class of transformations $\down{\alpha}$ is defined up to a translation, since $s(x)$ depends on an additive integration constant. Without loss of generality, the following \emph{canonical election} will be employed as a standard choice.
\begin{equation}\label{eq:can_change}
s(x)=\left\{\begin{array}{ll}\frac{\pdf^{2-\alpha}(x)}{\alpha-2}, & {\rm for} \ \alpha\in\Rset\setminus\{2\},\\
-\ln\,\pdf(x), & {\rm for} \ \alpha=2,\end{array}\right.
\end{equation}	
Many properties of the down transformation are established in previous works by the authors \cite{IP2025, IP2025(b)}. Let us recall here only that the canonical election implies that $(\alpha-2)s(x)\geqslant0$ for any $x$ in the support of $f$.

\medskip

\noindent \textbf{The up transformation.} Contrary to the down transformation, the up transformation is applicable to any probability density function. We give below its precise definition.
\begin{definition}\label{def:up}
Let $\pdf: \supp\longrightarrow \mathbb R^+$ be a probability density function with $\supp=(x_i,x_f)$. For $\alpha\in\Rset\setminus\{2\}$, the up transformation $\up{\alpha}$ is defined as
\begin{equation}\label{eq:up}
	f^{\uparrow}_\alpha(u)=\up{\alpha}[\pdf(x)](u)=|(\alpha-2)x(u)|^\frac{1}{2-\alpha},\quad u'(x)=-|(\alpha-2)x|^\frac{1}{\alpha-2}f(x),
\end{equation}
while for $\alpha=2$ the up transformation $\up{2}$ is defined as
\begin{equation}\label{eq:up2}
f^{\uparrow}_2(u)=\up{2}[\pdf(x)](u)=e^{-x(u)},\quad u'(x)=-e^xf(x).
\end{equation}
\end{definition}
We note that the definition of $u(x)$ in Eqs. \eqref{eq:up} and~\eqref{eq:up2} is taken up to a translation; for simplicity, the canonical choice is the primitive
$$
u(x)=\int_x^{x_f}|(\alpha-2)x|^\frac{1}{\alpha-2}f(x)dx, \quad \alpha\neq2,
$$
and the similar one for $\alpha=2$, where $x_f\in\Rset\cup\{\infty\}$ is the upper edge of the support of the domain of $f$, whenever this integral is finite. In the contrary case, we can employ any intermediate point $x_0$ in $(x_i,x_f)$ if the integral is divergent at its end.

The next result shows that the down and up transformations are mutually inverse.
\begin{proposition}\label{prop:inv}
Let $\pdf$ a probability density and let $\adown=\down{\alpha}[f]$ and $\aup=\up{\alpha}[f]$ be its $\alpha$-order down and up transformations. Then, up to a translation,
	\begin{equation}
	\pdf=\up{\alpha}[\adown]=\down{\alpha}[\aup].
	\end{equation}
that is, $\down{\alpha}\up{\alpha}=\up{\alpha}\down{\alpha}=\mathbb I$, where $\mathbb I$ denotes the identity operator.
\end{proposition}

\medskip

\noindent \textbf{Relation between up/down transformations, moments and entropies.} We next gather in the following technical result several identities showing how the up and down transformations relate to some of the classical informational functionals. Some of these properties will be employed thoughout the paper.
\begin{lemma}\label{lem:MEF}
Let $\pdf$ be a probability density and $\aup$ and $\adown$ its up/down transformations (assuming, whenever needed, that $f$ is derivable and decreasing). Then, if $\alpha\in\Rset\setminus\{2\}$, the following equalities hold true:
	\begin{equation}\label{eq:ME}
	\sigma_p[\adown]=\frac{N_{1+(2-\alpha)p}^{\alpha-2}[\pdf]}{|2-\alpha|},\quad\text{or equivalently,}\quad N_{\lambda}[\aup]=\left(|2-\alpha| \sigma_{\frac{\lambda-1}{2-\alpha}}[\pdf]\right)^{\frac{1}{\alpha-2}},
	\end{equation}
and
	\begin{equation}\label{eq:EF}
	N_{\lambda}[\adown]=\phi_{1-\lambda,2-\alpha}^{2-\alpha}[\pdf],\quad\text{or equivalently,}\quad  \phi_{p,\beta}[f^{\uparrow}_{2-\beta}]=\left(N_{1-p}[\pdf]\right)^{\frac1{\beta}}.
	\end{equation}
For the Shannon entropy we have
	\begin{equation}\label{eq:SUD}
	S[\adown]=\alpha S[f]+\big \langle \log \left| f'\right|\big \rangle_f,\qquad S[\aup]=\frac1{2-\alpha}\left\langle \log |x|\right\rangle_f+\frac{\log|2-\alpha|}{2-\alpha}.
	\end{equation}
For $\alpha=2$, we have the following equalities:
\begin{equation}\label{eq:Sp_down2}
\sigma_p[\down{2}[\pdf]]=\left[\int_{\Rset}f(x)|\ln\,f(x)|^pdx\right]^{\frac{1}{p}},
\end{equation}
respectively
\begin{equation}\label{eq:Ren_down2}
N_\lambda[\down{2}[\pdf]]=\lim_{\widetilde \lambda\to0}\phi_{1-\lambda,\widetilde \lambda}[f]^{\widetilde\lambda}\equiv F_{1-\lambda,0}[f]^\frac{1}{1-\lambda}.
\end{equation}
\end{lemma}

\medskip

\noindent \textbf{Upper-moments and down-Fisher measures.} We finally recall the construction of two families of informational functionals that have been defined and explored in \cite{IP2025(b)} with the aid of the up and down transformations. The first one is actually deduced by applying the absolute $p$-th moments to an up transformed density, as shown in \cite[Section 3.1]{IP2025(b)}.
\begin{definition}[Upper-moments]\label{def:hypermom}
Let $p\in\Rset$ and $f:\supp\rightarrow \Rset^+$ be a probability density. Then, for $\alpha\neq2$, we define the $(p,\alpha)-$\textit{upper-moments}, $ M_{p,\alpha}[\pdf],$ as
\begin{equation}\label{eq:hyper}
 M_{p,\alpha}[\pdf]=\int_{\Rset}\left|\int_x^{x_f} |(\alpha-2) v|^\frac1{\alpha-2}f(v)dv \right|^p\pdf(x)dx.
\end{equation}
For $\alpha=2$, the $(p,2)-$\textit{upper-moments} $ M_{p,2}[\pdf]$ is defined as
\begin{equation}\label{eq:hyper2}
M_{p,2}[\pdf]=\int_{\Rset}\left| \int_x^{x_f} e^vf(v)dv \right|^p\pdf(x)dx.
\end{equation}
Finally, for $\alpha\neq2$ and $p\neq0$ we define the $(p,\alpha)-$\textit{upper-deviation} as
\begin{equation}\label{eq:hyperdev}
m_{p,\alpha}[f]=M_{p,\alpha}[f]^\frac{\alpha-2}p.
\end{equation}
\end{definition}
The next family of measures is derived by applying the $(p,\lambda)$-Fisher information to a down transformed density, as seen in \cite[Section 3.3]{IP2025(b)}.
\begin{definition}[Down-Fisher measures]\label{def:sub-fisher}
Let $f$ be a differentiable up to second order and monotone probability density function and $p$, $q$, $\lambda$ three real numbers such that $p\neq q$. We introduce the \emph{down-Fisher measure}
\begin{equation}\label{eq:sub-fisher}
\varphi_{p,q,\lambda}[\pdf]=\int_{\Rset} f(v)^{1+p(\lambda-2)}|f'(v)|^{q}\left|\frac{p\lambda}{p-q}-\frac{\pdf(v) \pdf''(v)}{(\pdf'(v))^2}\right|^p\,dv.
\end{equation}	
\end{definition}
The following identity motivated the introduction of the down-Fisher measures. As proved in \cite[Lemma 3.1]{IP2025(b)}, in the same conditions as in Definition \ref{def:sub-fisher} and for any $\alpha\in\Rset$, we have
\begin{equation}\label{eq:gener_down_Fisher}
F_{p,\lambda}[f^{\downarrow}_{\alpha}]=\varphi_{p,p(1-\lambda),\alpha\lambda}[f].
\end{equation}
We end this section by recalling an inequality between down-Fisher measures with different parameters, that is very useful in the construction of measures of statistical complexity. Given $p$, $q\in\Rset\setminus\{0\}$ such that $p>q$, $\lambda\in\Rset$, $s\in\Rset\setminus\{p\}$ and $f$ a probability density function as in Definition \ref{def:sub-fisher}, we have proved in \cite[Theorem 3.1]{IP2025(b)} that
\begin{equation}\label{eq:order_gener_dF}
\varphi_{p,s,\lambda}^{\frac{1}{p}}[f]\geqslant\varphi_{q,\frac{qs}{p},\lambda}^{\frac{1}{q}}[f],
\end{equation}
and the equality is achieved for the minimizer
$$
f_{\rm min}:=\mathfrak{U}_{\frac{p+s}{p}}\mathfrak{U}_{\frac{p\lambda}{p-s}}[u], \quad u(x)=\frac{1}{x_f-x_i}, \quad x\in[x_i,x_f].
$$

\subsection{Previous informational inequalities}\label{sec:prel_ineq}

We conclude this section of preliminary facts, notions and results by listing a number of informational inequalities, both classical and very recently obtained, that will be used in the proofs.

\medskip

\noindent \textbf{The biparametric Stam inequality.} This is an inequality establishing that the product of the Rényi entropy power and the $(p,\lambda)$-Fisher information is bounded from below. More precisely, given $p\in[1,+\infty)$, $p^*=\frac{p}{p-1}$ (with the convention $p^*=\infty$ for $p=1$) and $\lambda>\frac1{1+p^*}$, the following inequality
\begin{equation}\label{ineq:bip_Stam}
\phi_{p,\lambda}[\pdf] \, N_{\lambda}[\pdf] \: \geqslant \: \phi_{p,\lambda}[\gauss_{p,\lambda}] \, N_{\lambda}[\gauss_{p,\lambda}]\equiv \kappa^{(1)}_{p,\lambda},
\end{equation}
holds true for any absolutely continuous probability density $\pdf$, according to~\cite{Lutwak2005, Bercher2012a}. Observe that the inequality Eq. \eqref{ineq:bip_Stam} is saturated by the minimizers $\gauss_{p,\lambda}$ appearing in the right hand side, known as generalized Gaussians \cite{Lutwak2005}, $q-$Gaussians \cite{Bercher2012a} or stretched deformed Gaussians \cite{Zozor2017}. For $p>1$, these minimizers are given by
\begin{equation}\label{def:g_plambda}
\gauss_{p,\lambda}(x) \, = \, \frac{a_{p,\lambda}}{\exp_\lambda\left( |x|^{p^*} \right)}
\, = \, a_{p,\lambda}\, \exp_{2-\lambda}\left( - |x|^{p^*} \right),
\end{equation}
where $\exp_\lambda$ is the generalized Tsallis exponential
\begin{equation}\label{def:q-exp}
\exp_\lambda(x) = \left( 1 + (1 - \lambda) \, x \right)_+^\frac1{1-\lambda}, \ \ \lambda \ne 1, \qquad \exp_1(x) \: \equiv \: \lim_{\lambda \to 1} \, \exp_\lambda(x) \: = \: \exp(x),
\end{equation}
and $a_{p,\lambda}$ has an explicit value which we omit here (see \cite{IP2025}). Observe that the definition \eqref{def:g_plambda} and the inequality \eqref{ineq:bip_Stam} can be extended to exponents $p^* \in (0,1)$ (that is, $p \in (-\infty,0)$). In the special case $p^*=p=0$ and $\lambda>1$, the Stam inequality also applies, but its minimizer is given by
$$
\gauss_{0,\lambda} = a_{0,\lambda}(-\log|x|)_+^{\frac1{\lambda-1}}, \quad a_{0,\lambda} = \frac1{2\Gamma\left(\frac{\lambda}{\lambda-1}\right)}.
$$
The limit $p\to 1$ entails $p^*\to\infty$ and then $\gauss_{1,\lambda}$ becomes a constant density over a unit length support, while in the limit $p\to\infty$ the inequality also holds true by taking instead of $\phi_{p,\lambda}^{\lambda}$ the essential supremum (that is, the $L^{\infty}$ norm) as observed in~\cite{Lutwak2004}.

\medskip

\noindent \textbf{The moment-entropy inequality} is an informational inequality relating the R\'enyi power entropy and the moments $\sigma_p$. More precisely, when
\begin{equation}\label{eq:param_clas}
p^*\in[0,\infty ), \quad {\rm and} \quad \lambda>\frac1{1+p^*},
\end{equation}
it was proved in~\cite{Lutwak2004, Lutwak2005, Bercher2012} that, for any probability density function $f$, we have
\begin{equation}\label{ineq:bip_E-M}
\frac{\sigma_{p^*}[\pdf]}{N_{\lambda}[\pdf]} \, \geqslant \, \frac{\sigma_ {p^*}[\gauss_{p,\lambda}]}{N_{\lambda}[\gauss_{p,\lambda}]}\equiv K^{(0)}_{p,\lambda},
\end{equation}
and the minimizers of Eq. \eqref{ineq:bip_E-M} are the same deformed Gaussians $g_{p,\lambda}$ as for the generalized biparametric Stam inequality Eq.~\eqref{ineq:bip_Stam}. We have extended the inequality \eqref{ineq:bip_E-M} in our previous work \cite[Theorem 5.2]{IP2025} to a mirrored range of parameters. More precisely, it was established therein that, if
\begin{equation}\label{eq:param_mir}
\lambda<0,\quad\sign\left(\frac{\lambda-1}{ \lambda}+ p^*\right)=\sign\left(1-p^*\right),
\end{equation}
then, for any continuously differentiable density function $f$, the following mirrored moment-entropy inequality holds true:
\begin{equation}\label{ineq:bip_E-M_mirrored}
\left(\frac{\sigma_{p^*}[\pdf]}{N_{\lambda}[\pdf]}\right)^{p^*-1}\geqslant\left(\frac{\sigma_{p^*}[\overline g_{p,\lambda}]}{N_{\lambda}[\overline g_{p,\lambda}]}\right)^{p^*-1}\equiv \kappa^{(0)}_{p,\lambda},
\end{equation}
where $\overline g_{p,\lambda}=g_{1-\lambda,1-p}$ and $\kappa^{(0)}_{p,\lambda}=|p^*-1|^{p^*-1}\kappa^{(1)}_{1-\lambda,\frac 1{1-p}}$. Throughout the paper we adopt the following notation
\begin{equation}\label{eq:not_min}
	\widetilde{g}_{p,\lambda}=\begin{cases}
g_{p,\lambda},\quad\lambda>0,\\
\overline g_{p,\lambda},\quad \lambda<0.
	\end{cases}
\end{equation}

\medskip

\noindent \textbf{The Cram\'er-Rao inequality} is an informational inequality relating the moments $\sigma_p$ and the $(p,\lambda)$-Fisher information and which follows trivially from \eqref{ineq:bip_E-M} and \eqref{ineq:bip_Stam} by multiplication. However, we state it below since it will be employed in the paper. When $p^*\in[0,\infty )$ and $\lambda>\frac1{1+p^*}$, we have
\begin{equation}\label{ineq:CR}
\phi_{p,\lambda}[\pdf]\sigma_{p^*}[\pdf]\geqslant \phi_{p,\lambda}[\gauss_{p,\lambda}]\sigma_{p^*}[\gauss_{p,\lambda}]\equiv K^{(0)}_{p,\lambda}\kappa^{(1)}_{p,\lambda},
\end{equation}
for any absolutely continuous probability density function.

\medskip

\noindent \textbf{The extended (triparametric) Stam inequality}. This is a generalized inequality, valid both in the classical domain of parameters and in a mirrored domain, established first in \cite{Zozor2017} and then extended by the authors in the previous work \cite[Theorem 5.1]{IP2025} with the help of the up and down transformations. Let $p\geqslant 1$ and $q$ be such that
\begin{equation}\label{eq:sign_cond1}
\sign\left(p^*q + \lambda-1\right)=\sign\left(q+1-\lambda\right)\neq0.
\end{equation}
Then, the following generalized Stam inequality holds true for $f: \Rset\mapsto\Rset^+$ absolutely continuous if $1+q-\lambda>0$ or for $f:(x_i,x_f)\mapsto \Rset^+$ absolutely continuous on $(x_i,x_f)$ if $1+q-\lambda<0$:
\begin{equation}\label{ineq:trip_Stam_extended}
	\left(\phi_{p,q}[\pdf] \, N_{\lambda}[\pdf] \right)^{\theta(q,\lambda)}\: \geqslant \:	\left(\phi_{p,q}[\widetilde g_{p,q,\lambda}] \, N_{\lambda}[\widetilde g_{p,q,\lambda}] \right)^{\theta(q,\lambda)}\equiv \kappa_{p,q,\lambda}^{(1)},
\end{equation}
where $\theta(q,\lambda)=1+q-\lambda$ and $\widetilde g_{p,q,\lambda}$ is defined in \cite[Section 5]{IP2025}. Moreover, for $p<1$ and $q, \lambda$ such that
\begin{equation}\label{eq:sign_cond2}
\sign(p^*q+\lambda-1)=\sign(q+\lambda-1)\neq0, \quad \sign(\lambda-1)=\sign(q)\neq0.
\end{equation}
the inequality ~\eqref{ineq:trip_Stam_extended} holds true with $\theta(q,\lambda)=-q$, provided that $f:\Rset \mapsto\Rset^+$ is continuously differentiable with $f'<0$. The optimal constants and the minimizers $\widetilde g_{p,q,\lambda}$ are known, but we omit the (rather technical) expressions here and we refer the interested reader to \cite[Section 5]{IP2025} for them.

\medskip

\noindent \textbf{Cumulative inequalities.} It follows from \cite[Theorem 1, Section 5.1]{Puertas2025} that we have the following moment-entropy like inequality involving cumulative moments: given $p$ such that $p^*\geqslant0$, $\beta$, $\lambda$ such that
\begin{equation}\label{eq:cond_cumul}
\beta>\max\left\{\lambda-1,\frac{1-\lambda}{p}\right\},
\end{equation}
for any probability density function $f$ we have
\begin{equation}\label{ineq:E-M-cumul}
\frac{\sigma_{p^*,1+\beta-\lambda}[f]}{N_{\lambda}[f]}\geqslant\frac{\sigma_{p^*,1+\beta-\lambda}[g_{p,\beta,\lambda}]}{N_{\lambda}[g_{p,\beta,\lambda}]}
\equiv K_{p,\beta,\lambda}^{(0)},
\end{equation}
where both the explicit form of the minimizers $g_{p,\beta,\lambda}$ and the optimal constant $K_{p,\beta,\lambda}^{(0)}$ are given in \cite{Puertas2025}. An inspection of the proof of~\cite[Theorem 3, Section 5.1]{Puertas2025} shows that the condition for $f$ to be continuously differentiable can be removed.

\section{Biparametric up and down transforms: definitions}\label{sec:transf}

We gather in this section the definitions and basic properties of the objects that we introduce and analyze in this work: biparametric up and down transformations and several informational functionals and measures whose definition is motivated by the new class of transformations.

\subsection{Biparametric up and down transformations}

We next define the biparametric family of transformations that will be at the core of the present work.
\begin{definition}\label{def:bip_down}
Let $f:(x_i,x_f)\mapsto\Rset$ be a derivable probability density function such that $f'(x)<0$ for any $x\in(x_i,x_f)$, where $-\infty<x_i<x_f\leqslant\infty$. For $\alpha$, $\beta\in\Rset$ arbitrary, we define the \emph{biparametric down transformation} by
\begin{equation}\label{eq:bip_down}
\mathfrak{D}_{\alpha,\beta}[f](s):=\frac{f(x)^{\alpha}}{|f'(x)|^{\beta}}, \quad s'(x)=f(x)^{1-\alpha}|f'(x)|^{\beta}.
\end{equation}
\end{definition}
Let us notice that $\mathfrak{D}_{\alpha,\beta}[f]$ is defined up to a translation, since the change of the independent variable in \eqref{eq:bip_down} depends on an integration constant. Moreover, it is immediate to see by a straightforward change of variable that $\mathfrak{D}_{\alpha,\beta}[f]$ is a probability density function as well. In order to simplify the notation, we will employ throughout the paper the alternative notation
\begin{equation}\label{eq:not_down}
f^{\Downarrow}_{\alpha,\beta}(s):=\mathfrak{D}_{\alpha,\beta}[f](s).
\end{equation}
Recalling the down and the differential-escort transformations defined in Sections \ref{subsec:escort} and \ref{subsec:updown}, we observe that
\begin{equation}\label{eq:escort_down}
\mathfrak{D}_{\alpha,\beta}[f]=\mathfrak{D}_{\frac{\alpha}{\beta}}[f]^{\beta}=\mathfrak{E}_{\beta}\circ\mathfrak{D}_{\frac{\alpha}{\beta}}[f],
\end{equation}
for any $\beta\neq0$. Let us remark at this point that, due to the application of the $\mathfrak{D}_{\frac{\alpha}{\beta}}$, we have
$$
\left(\frac{\alpha}{\beta}-2\right)s(x)\geqslant0,
$$
for any $x$ in the suppprt of $f$. Moreover, we trivially have
$$
\mathfrak{D}_{\alpha,0}=\mathfrak{E}_{\alpha}, \quad \mathfrak{D}_{\alpha,1}=\mathfrak{D}_{\alpha}.
$$
The latter equalities show that, in particular, the biparametric family $\mathfrak{D}_{\alpha,\beta}$ provides an interpolation between the differential-escort and the down transformations and thus is expected to further enrich the already very interesting structure of informational inequalities developed in previous papers (see, for example, \cite{Zozor2017, Puertas2025, IP2025, IP2025(b)}). The alternative writing \eqref{eq:escort_down} allows us to define the \emph{biparametric up transformation} as the inverse of the biparametric down transformation.
\begin{definition}\label{def:bip_up}
Let $f:(x_i,x_f)\mapsto\Rset$ be a probability density function. We define the \emph{biparametric up transformation} by
\begin{equation}\label{eq:bip_up}
\mathfrak{U}_{\alpha,\beta}[f]:=\mathfrak{U}_{\frac{\alpha}{\beta}}[\mathfrak{E}_{\beta^{-1}}[f]],
\end{equation}
for any $\beta\in\Rset\setminus\{0\}$, while
$$
\mathfrak{U}_{\alpha,0}[f]=\mathfrak{E}_{\alpha^{-1}}[f].
$$
\end{definition}
We readily deduce from \eqref{eq:escort_down} and Proposition \ref{prop:inv} that
$$
\mathfrak{D}_{\alpha,\beta}\mathfrak{U}_{\alpha,\beta}=
\mathfrak{E}_{\beta}\mathfrak{D}_{\frac{\alpha}{\beta}}\mathfrak{U}_{\frac{\alpha}{\beta}}\mathfrak{E}_{\beta^{-1}}=\mathbb I
$$
and in the same way $\mathfrak{U}_{\alpha,\beta}\mathfrak{D}_{\alpha,\beta}=\mathbb I$, where $\mathbb I$ designs the identity operator. Thus, $\mathfrak{D}_{\alpha,\beta}$ and $\mathfrak{U}_{\alpha,\beta}$ are mutually inverse. We will also use throughout the paper the simplified notation
\begin{equation}\label{eq:not_up}
f^{\Uparrow}_{\alpha,\beta}(u):=\mathfrak{U}_{\alpha,\beta}[f](u).
\end{equation}
\begin{remark}
In strong contrast to the biparametric down transformation, the biparametric up transformed density does not have a ``pleasant" expression in terms of the original probability density function. Indeed, given $\alpha$, $\beta\in\Rset$ such that $\beta\neq0$ and $\alpha\neq2\beta$, one can compute that
$$
\mathfrak{U}_{\alpha,\beta}[f](u)=\mathfrak{U}_{\frac{\alpha}{\beta}}[f_{\frac{1}{\beta}}(y)](u)
=\left|\frac{\alpha-2\beta}{\beta}y(u)\right|^{\frac{\beta}{2\beta-\alpha}},
$$
where the change of the independent variable is given by
\begin{equation*}
\begin{split}
u(y)&=-\int_{y}^{y_f}\left|\left(\frac{\alpha}{\beta}-2\right)y\right|^{\frac{\beta}{\alpha-2\beta}}f_{\beta^{-1}}(y)\,dy\\
&=-\int_{x}^{x_f}\left|\left(\frac{\alpha}{\beta}-2\right)\int_0^xf^{\frac{\beta-1}{\beta}}(t)\,dt\right|^{\frac{\beta}{\alpha-2\beta}}f(x)\,dx\\
&=-\left|\frac{\alpha-2\beta}{\beta}\right|^{\frac{\beta}{\alpha-2\beta}}\int_x^{x_f}\left|\int_0^xf^{\frac{\beta-1}{\beta}}(t)\,dt\right|^{\frac{\beta}{\alpha-2\beta}}f(x)\,dx.
\end{split}
\end{equation*}
This is why, throughout the paper, the calculations based on an application of the biparametric up transformation will employ the definition by composition \eqref{eq:bip_up}.
\end{remark}
We next calculate the derivative of the biparametric down-transformed density. This calculation will be useful in the sequel, when working with the Shannon entropy applied to a biparametric down-transformed density. Recalling the notation \eqref{eq:not_down} and the fact that $f$ is assumed to be decreasing and derivable up to second order, we have
\begin{equation}\label{eq:deriv_down}
\begin{split}
\frac{df^{\Downarrow}_{\alpha,\beta}}{ds}&=\frac{d}{ds}\left[f(x(s))^{\alpha}(-f'(x(s)))^{-\beta}\right]\\
&=-\alpha f(x)^{2\alpha-2}(-f'(x))^{1-2\beta}+\beta f(x)^{2\alpha-1}(-f'(x))^{-2\beta-1}f''(x)\\
&=\beta f(x)^{2\alpha-2}|f'(x)|^{-2\beta+1}\left(\frac{f(x)f''(x)}{(f'(x))^2}-\frac{\alpha}{\beta}\right).
\end{split}
\end{equation}
We remark that the down transformed density $f^{\Downarrow}_{\alpha,\beta}$ is not necessarily a decreasing function. Indeed, we deduce from \eqref{eq:deriv_down} that, in order to be able to apply once more the biparametric down transformation to $f^{\Downarrow}_{\alpha,\beta}$, the following condition
\begin{equation}\label{cond:down_twice}
\sup\limits_{x\in\Rset}\frac{f(x)f''(x)}{(f'(x))^2}<\frac{\alpha}{\beta}
\end{equation}
has to be fulfilled. An immediate application of Eq. \eqref{eq:deriv_down} is the following result involving the Shannon entropy.
\begin{proposition}\label{prop:Shannon}
For any $\alpha$, $\beta\in\Rset$ and for any differentiable and decreasing probability density function $f$ we have
\begin{equation}\label{eq:Shan_down}
S[f^{\Downarrow}_{\alpha,\beta}]=\alpha S[f]+\beta\int_{\Rset}f(x)\log|f'(x)|\,dx.
\end{equation}
Moreover, for any $(\alpha_1,\beta_1,\alpha_2,\beta_2)\in\Rset^4$ such that $\beta_1\neq0$ and for any density $f$ differentiable up to the second order, decreasing and such that the condition \eqref{cond:down_twice} is satisfied with $(\alpha,\beta)=(\alpha_1,\beta_1)$, we have
\begin{equation}\label{eq:Shan_down_down}
\begin{split}
S[\mathfrak{D}_{\alpha_2,\beta_2}[f^{\Downarrow}_{\alpha_1,\beta_1}]]&=(\alpha_1\alpha_2+2\beta_2-2\alpha_1\beta_2)S[f]\\
&+(\beta_1\alpha_2+\beta_2-2\beta_1\beta_2)\left\langle\log|f'|\right\rangle_f\\
&+\beta_2\left\langle\log\left|\frac{ff''}{(f')^2}-\frac{\alpha_1}{\beta_1}\right|\right\rangle_f+\beta_2\log\,\beta_1.
\end{split}
\end{equation}
\end{proposition}
\begin{proof}
We start from the definition of the Shannon entropy and calculate
\begin{equation*}
\begin{split}
S[f^{\Downarrow}_{\alpha,\beta}]&=-\int_{\Rset}f^{\Downarrow}_{\alpha,\beta}(s)\log\,f^{\Downarrow}_{\alpha,\beta}(s)\,ds
=-\int_{\Rset}f(x)\log(f(x)^{\alpha}|f'(x)|^{-\beta})\,dx\\
&=-\alpha\int_{\Rset}f(x)\log\,f(x)\,dx+\beta\int_{\Rset}f(x)\log|f'(x)|\,dx=\alpha S[f]+\beta\int_{\Rset}f(x)\log|f'(x)|\,dx.
\end{split}
\end{equation*}
In order to establish Eq. \eqref{eq:Shan_down_down}, we apply the already proved identity Eq. \eqref{eq:Shan_down} twice, as follows:
\begin{equation*}
\begin{split}
S[\mathfrak{D}_{\alpha_2,\beta_2}[f^{\Downarrow}_{\alpha_1,\beta_1}]]&=\alpha_2S[f^{\Downarrow}_{\alpha_1,\beta_1}]
+\beta_2\int_{\Rset}f^{\Downarrow}_{\alpha_1,\beta_1}(s)\log\left|\frac{d}{ds}f^{\Downarrow}_{\alpha_1,\beta_1}(s)\right|\,ds\\
&=\alpha_2(\alpha_1S[f]+\beta_1\int_{\Rset}f(x)\log|f'(x)|\,dx)\\
&+\beta_2\int_{\Rset}\log\left[\beta_1f(x)^{2\alpha_1-2}|f'(x)|^{1-2\beta_1}\left|\frac{f(x)f''(x)}{(f'(x))^2}-\frac{\alpha_1}{\beta_1}\right|\right]f(x)\,dx\\
&=\alpha_1\alpha_2S[f]+\beta_1\alpha_2\left\langle\log|f'|\right\rangle_f+\beta_2(2-2\alpha_1)S[f]+\beta_2(1-2\beta_1)\left\langle\log|f'|\right\rangle_f\\
&+\beta_2\left\langle\log\left|\frac{ff''}{(f')^2}-\frac{\alpha_1}{\beta_1}\right|\right\rangle_f+\beta_2\log\,\beta_1,
\end{split}
\end{equation*}
from which Eq. \eqref{eq:Shan_down_down} readily follows by gathering similar terms.
\end{proof}
We end this section by an easy but very useful result in the proof of the forthcoming informational inequalities. It shows that the R\'enyi entropy power applied to a biparametric down transformed density is related to the Fisher information, while the R\'enyi entropy power applied to a biparametric up transformed density produces a cumulative moment.
\begin{lemma}\label{lem:ent_down_up}
For any $\alpha$, $\beta\in\Rset$ such that $\beta\neq0$ and $\alpha\neq2\beta$, and $\lambda\in\Rset\setminus\{1\}$, we have
\begin{equation}\label{eq:Ren_down}
N_{\lambda}[f^{\Downarrow}_{\alpha,\beta}]=\phi_{(1-\lambda)\beta,2-\frac{\alpha}{\beta}}^{2\beta-\alpha}[f].
\end{equation}
In the same conditions, we also have
\begin{equation}\label{eq:Ren_up}
N_{\lambda}[f^{\Uparrow}_{\alpha,\beta}]^{1-\lambda}=\left|\frac{2\beta-\alpha}{\beta}\right|^{\frac{\beta(1-\lambda)}{\alpha-2\beta}}
\mu_{\frac{\beta(\lambda-1)}{2\beta-\alpha},\frac{1}{\beta}}[f]
\end{equation}
\end{lemma}
\begin{proof}
By direct calculation, we have
\begin{equation}\label{eq:interm1}
\begin{split}
N_{\lambda}[f^{\Downarrow}_{\alpha,\beta}]&=\left[\int_{\Rset}(f(x)^{\alpha}|f'(x)|^{-\beta})^{\lambda-1}f(x)\,dx\right]^{\frac{1}{1-\lambda}}\\
&=\left[\int_{\Rset}f(x)^{1+(\lambda-1)\alpha}|f'(x)|^{(1-\lambda)\beta}\,dx\right]^{\frac{1}{1-\lambda}}.
\end{split}
\end{equation}
Introducing $(p,\overline{\lambda})$ such that
$$
1+(\lambda-1)\alpha=1+(\overline{\lambda}-2)p, \quad (1-\lambda)\beta=p,
$$
we easily derive that
$$
\overline{\lambda}=2+\frac{(\lambda-1)\alpha}{p}=2-\frac{\alpha}{\beta}.
$$
With this in mind, we continue the calculation started in \eqref{eq:interm1} to find
$$
N_{\lambda}[f^{\Downarrow}_{\alpha,\beta}]=\left[\phi_{p,\overline{\lambda}}^{p\overline{\lambda}}[f]\right]^{\frac{1}{1-\lambda}},
$$
which leads to \eqref{eq:Ren_down} after replacing $p$ and $\overline{\lambda}$ in terms of $\alpha$, $\beta$, $\lambda$. In order to derive \eqref{eq:Ren_up}, we recall the definition \eqref{def:cum_mom} of the cumulative moment and the one of the biparametric up transformation \eqref{def:bip_up} and find
\begin{equation*}
\begin{split}
N_{\lambda}^{1-\lambda}[f^{\Uparrow}_{\alpha,\beta}]&=N_{\lambda}^{1-\lambda}[\mathfrak{U}_{\frac{\alpha}{\beta}}\mathfrak{E}_{\beta^{-1}}[f]]\\
&=\left|\frac{2\beta-\alpha}{\beta}\right|^{\frac{\beta(1-\lambda)}{\alpha-2\beta}}
\sigma_{\frac{\beta(\lambda-1)}{2\beta-\alpha}}^{\frac{\beta(1-\lambda)}{\alpha-2\beta}}[\mathfrak{E}_{\beta^{-1}}[f]]\\
&=\left|\frac{2\beta-\alpha}{\beta}\right|^{\frac{\beta(1-\lambda)}{\alpha-2\beta}}\mu_{\frac{\beta(\lambda-1)}{2\beta-\alpha},\frac{1}{\beta}}[f],
\end{split}
\end{equation*}
completing the proof.
\end{proof}
Let us observe that the results in Proposition \ref{prop:Shannon} and Lemma \ref{lem:ent_down_up} reduce to the ones established in \cite[Lemma 3.1]{IP2025} if $\beta=1$, respectively $\beta_1=\beta_2=1$, when we apply the (standard) down transformation.

\subsection{Cumulative upper-moments}

Let $\alpha$, $\beta$, $p\in\Rset$ such that $\alpha\neq 2\beta$ and $\beta\neq0$. For a probability density function $f$, we define the \emph{cumulative upper-moment} by the following integral quantity:
\begin{equation}\label{def:cumul_upper}
M_{p,\alpha,\beta}[f]:=\left|\frac{\alpha-2\beta}{\beta}\right|^{\frac{p\beta}{\alpha-2\beta}}\int_{\Rset}\left|\int_x^{x_f}\left|\int_0^r f^{\frac{\beta-1}{\beta}}(t)\,dt\right|^{\frac{\beta}{\alpha-2\beta}}f(r)\,dr\right|^pf(x)\,dx.
\end{equation}
In the case when $\alpha=2\beta$, the definition has to be changed in order to include an exponential, similarly as in the critical case $\alpha=2$ in the definition of upper-moments \eqref{eq:hyper2}. More precisely, in this case we define:
\begin{equation}\label{def:cumul_upper2}
M_{p,2,\beta}[f]:=\int_{\Rset}\left|\int_x^{x_f}\exp\left\{\int_0^rf^{\frac{\beta-1}{\beta}}(t)\,dt\right\}f(r)\,dr\right|^pf(x)\,dx.
\end{equation}
Notice that, if $\beta=1$, the cumulative upper-moment $M_{p,\alpha,1}$ reduces to the upper-moment $M_{p,\alpha}$ introduced in \cite[Section 3.1]{IP2025(b)}, provided that $\alpha\neq2$. The following result motivates the definition of the cumulative upper-moment.
\begin{proposition}\label{prop:cum_upper}
In the previous conditions and notation, we have
\begin{equation}\label{eq:cum_upper}
\mu_p[f^{\Uparrow}_{\alpha,\beta}]=M_{p,\alpha,\beta}[f].
\end{equation}
\end{proposition}
\begin{proof}
Assume first that $\alpha\neq 2\beta$. Recalling the definitions of the upper-moments \eqref{eq:hyper} and of the biparametric up transformation \eqref{def:bip_up}, we deduce that
\begin{equation*}
\begin{split}
\mu_p[f^{\Uparrow}_{\alpha,\beta}]&=\mu_p[\mathfrak{U}_{\frac{\alpha}{\beta}}\mathfrak{E}_{\frac{1}{\beta}}[f]]=M_{p,\frac{\alpha}{\beta}}[f_{\beta^{-1}}]\\
&=\int_{\Rset}\left|\int_y^{y_f}\left|\left(\frac{\alpha}{\beta}-2\right)v\right|^{\frac{\beta}{\alpha-2\beta}}f_{\beta^{-1}}(v)\,dv\right|^pf_{\beta^{-1}}(y)\,dy\\
&=\int_{\Rset}\left|\int_x^{x_f}\left|\left(\frac{\alpha}{\beta}-2\right)\int_0^rf^{\frac{\beta-1}{\beta}}(t)\,dt\right|^{\frac{\beta}{\alpha-2\beta}}f(r)\,dr\right|^pf(x)\,dx,
\end{split}
\end{equation*}
which leads to \eqref{eq:cum_upper} after obvious manipulations. If now $\alpha=2\beta$, we repeat the previous calculation but employing the upper-moment $M_{p,2}$ defined in \eqref{eq:hyper2} as follows:
\begin{equation*}
\begin{split}
\mu_p[f^{\Uparrow}_{\alpha,\beta}]&=\mu_p[\mathfrak{U}_{2}\mathfrak{E}_{\frac{1}{\beta}}[f]]=M_{p,2}[f_{\beta^{-1}}]\\
&=\int_{\Rset}\left|\int_y^{y_f}e^vf_{\beta^{-1}}(v)\,dv\right|^pf_{\beta^{-1}}(x)\,dx\\
&=\int_{\Rset}\left|\int_x^{x_f}\exp\left\{\int_0^rf^{\frac{\beta-1}{\beta}}(t)\,dt\right\}f(r)\,dr\right|^pf(x)\,dx,
\end{split}
\end{equation*}
which corresponds exactly to the definition of $M_{p,\alpha,\beta}$ when $\alpha=2\beta$.
\end{proof}

\noindent \textbf{Remark.} Since H\"older's inequality implies that $\mu_p^{\frac{1}{p}}[f]\leqslant\mu_{q}^{\frac{1}{q}}[f]$ for any $0<p<q$, we infer from applying the previous inequality to a biparametric up transformed density $\mathfrak{U}_{\alpha,\beta}[f]$ that
\begin{equation}\label{eq:order_upper}
M_{p,\alpha,\beta}^{\frac{1}{p}}[f]\leqslant M_{q,\alpha,\beta}^{\frac{1}{q}}[f],
\end{equation}
if $0<p<q$.

\subsection{Down-moments}

The second new informational quantity introduced in this paper is the \emph{down-moment}. For any $a$, $b$, $p\in\Rset$, and for any differentiable probability density function $f$, it is defined as
\begin{equation}\label{def:gener_down_moment}
\Xi_{p,a,b}[f]=\int_\Rset ~\left|\int_x^{x_f}f(t)^{1-a}|f'(t)|^bdt \right|^pf(x)dx.
\end{equation}
Recalling the definition of the cumulative moments \eqref{def:cum_mom}, the previous informational functional (and the name we have chosen for it) is motivated by the following property.
\begin{proposition}\label{prop:gener_down_mom}
Let $(p,\gamma,\alpha,\beta)\in\Rset^4$. Then, for any differentiable and decreasing probability density function, we have
\begin{equation}\label{eq:gener_down_mom}
\mu_{p,\gamma}[f^{\Downarrow}_{\alpha,\beta}]=\Xi_{p,\alpha\gamma,\beta\gamma}[f].
\end{equation}
In particular,
\begin{equation}\label{eq:gener_down_mom2}
\mu_p[f^{\Downarrow}_{\alpha,\beta}]=\Xi_{p,\alpha,\beta}[f].
\end{equation}
Finally, letting $b=1$ and for $a\in\Rset\setminus\{2\}$, the down-moments reduce to entropies:
\begin{equation}\label{eq:gdm1}
\Xi_{p,a,1}[f]=|2-a|^{-p}N_{1+p(2-a)}^{p(a-2)}[f],
\end{equation}
for any probability density function $f$ such that $\lim\limits_{x\to x_f}f(x)=0$.
\end{proposition}
\begin{proof}
We employ the definition of the cumulative moment \eqref{def:cum_mom} to calculate:
\begin{equation*}
\begin{split}
\mu_{p,\gamma}[f^{\Downarrow}_{\alpha,\beta}]&=\int_{\Rset}\left|\int_0^{s(x)}[f^{\Downarrow}_{\alpha,\beta}]^{1-\gamma}(t)\,dt\right|^pf^{\Downarrow}_{\alpha,\beta}(s)\,ds\\
&=\int_{\Rset}\left|\int_0^x[f(t)^{\alpha}|f'(t)|^{-\beta}]^{-\gamma}f(t)\,dt\right|^pf(x)\,dx\\
&=\int_{\Rset}\left|\int_0^xf(t)^{1-\alpha\gamma}|f'(t)|^{\beta\gamma}\,dt\right|^pf(x)\,dx=\Xi_{p,\alpha\gamma,\beta\gamma}[f],
\end{split}
\end{equation*}
as claimed. Eq. \eqref{eq:gener_down_mom2} follows as an immediate particular case by letting $\gamma=1$ in Eq. \eqref{eq:gener_down_mom}. Letting now $b=1$ in the definition Eq. \eqref{def:gener_down_moment}, we observe that the interior integral can be calculated to find
\begin{equation*}
\begin{split}
\Xi_{p,a,1}[f]&=\int_{\Rset}\left|\int_x^{x_f}f^{1-a}(t)f'(t)\,dt\right|^pf(x)\,dx=\int_{\Rset}\left|\frac{f(x)^{2-a}}{2-a}\right|^pf(x)\,dx\\
&=|2-a|^{-p}\int_{\Rset}f(x)^{1+p(2-a)}\,dx=|2-a|^{-p}N_{1+p(2-a)}^{p(a-2)}[f].
\end{split}
\end{equation*}
\end{proof}

\noindent \textbf{Remarks. 1.} An equivalent formulation of Eq. \eqref{eq:gener_down_mom} is that
$$
\Xi_{p,\alpha\gamma,\beta\gamma}[f^{\Uparrow}_{\alpha,\beta}]=\mu_{p,\gamma}[f].
$$
We have noticed that the down-moments reduce to entropies if letting $b=1$ in Eq. \eqref{def:gener_down_moment}. Another interesting particular case is when $b=0$ and $a=1$ in Eq. \eqref{def:gener_down_moment}, in which case the down-moments reduce to $\langle|x-x_f|^p\rangle_f$. If, furthermore, $x_f=0$, we are left with the standard $p$-th moments $\mu_p[f]$.

\medskip

\noindent \textbf{2.} We infer from \eqref{eq:gener_down_mom2} that the down-moments satisfy the following order relation:
\begin{equation}\label{eq:order_gener_dM}
\Xi_{p,\alpha,\beta}^{\frac{1}{p}}[f]\leqslant\Xi_{q,\alpha,\beta}^{\frac{1}{q}}[f],
\end{equation}
provided $0<p<q$.

\section{Informational inequalities}\label{sec:inequalities}

Taking as starting point the already established informational inequalities recalled in Section \ref{sec:prel_ineq}, we derive below a number of more general inequalities involving both classical informational functionals and the new ones introduced in the previous section.

\subsection{Cumulative moment-entropy inequality in the mirrored domain}

We begin with an inequality extending the cumulative moment-entropy inequality \eqref{ineq:E-M-cumul} to a mirrored domain of parameters, which can be also seen as a generalization to cumulative moments of the mirrored moment-entropy inequality \eqref{ineq:bip_E-M_mirrored}.
\begin{theorem}
	Let $(p,\delta,\lambda)\in\Rset^3$ be such that $\delta\neq0$ and the following sign conditions are satisfied:
\begin{equation}\label{eq:cond_mir_cumul}
\sign(1-\delta-\lambda)=\sign\,\delta\neq0, \quad \sign(1-p^*)=\sign\left(\frac{\lambda-1}{\delta+\lambda-1}+p^*\right)\neq0.
\end{equation}
Then we have
\begin{equation}\label{ineq:cum_ent_mir}
\left(\frac{\sigma_{p^*,\delta}[f]}{N_{\lambda}[f]}\right)^{\delta(p^*-1)}\geqslant\kappa_{p,\frac{\delta+\lambda-1}{\delta}}^{(0)},
\end{equation}	
for any continuously differentiable probability density function $f$, with the additional condition $f>0$ on its support if $\delta<0$. The minimizing densities are given by
\begin{equation}\label{eq:min_cum_EM}
f_{\rm min}(x):=\mathfrak{E}_{\frac{1}{\delta}}[g_{1-\lambda,1-p}](x)=k\begin{cases}
\left[\cos_{\frac{\delta p}{1-\delta},\frac{\lambda-1}\lambda}(y)\right]_+^\frac1{1-\delta},\quad p<0, \\[3mm]
\left[\cosh_{\frac{\delta p}{1-\delta},\frac{\lambda-1}\lambda}(y)\right]_+^\frac1{1-\delta},\quad p>0, \\
\end{cases}
\end{equation}
where $k$ is the normalization constant.
\end{theorem}
\begin{proof}
	We start from the moment-entropy inequality in the mirrored domain of parameters \eqref{ineq:bip_E-M_mirrored}, which we apply to a differential-escort transformed density $\mathfrak{E}_{\delta}[f]$. Note that the positivity conditions ensures that $\mathfrak{E}_{\delta}[f]$ is continuously differentiable even when $\delta<0$. From the properties of the differential-escort transformation~\eqref{eq:sigmaescort} and~\eqref{eq:Renyiescort} we find
$$
\kappa_{p,\lambda}^{(0)}\leqslant\left(\frac{\sigma_{p^*}[\mathfrak{E}_{\delta}[f]]}{N_{\lambda}[\mathfrak{E}_{\delta}[f]]}\right)^{p^*-1}
=\left(\frac{\sigma_{p^*,\delta}[f]}{N_{1+(\lambda-1)\delta}[f]}\right)^{\delta(p^*-1)}.
$$
We reparametrize
$$
\lambda':=1+(\lambda-1)\delta, \quad {\rm or, \ equivalently,} \quad \lambda=1+\frac{\lambda'-1}{\delta}.
$$
With this notation, it is easy to check that the condition \eqref{eq:param_mir} is written as \eqref{eq:cond_mir_cumul}, completing the proof of \eqref{ineq:cum_ent_mir} after dropping the primes from $\lambda'$. Since the minimizers of the mirrored entropy-moment inequality \eqref{ineq:bip_E-M_mirrored} are given by $g_{1-\lambda,1-p}$, then the minimizing densities satisfy
$$
\mathfrak E_{\delta}[f_{\rm min}]=g_{1-\lambda,1-p},
$$
and thus are obtained by applying the differential-escort transformation $\mathfrak{E}_{\frac{1}{\delta}}$ to the previous equality, leading to \eqref{eq:min_cum_EM} by taking into account~\cite[Lemma 3]{Puertas2025}.
\end{proof}

\noindent \textbf{Remark.} The same inequality \eqref{ineq:cum_ent_mir} can be derived from the triparametric Stam inequality \eqref{ineq:trip_Stam_extended} in the classical domain, by applying it to a biparametric up transformed density, in a similar manner as in the proof of \cite[Theorem 5.2]{IP2025}. This is equivalent to the proof we gave above, since the biparametric up transformation is the composition of a one-parameter up transformation and a differential-escort transformation, by definition.

\subsection{A down-moment--Fisher inequality}

The first one is an inequality involving the down-moment and the Fisher information.
\begin{theorem}\label{th:stam-like}
Let $(p,q,r,\lambda)\in\Rset^4$ such that, $p^*\geqslant0$, $\lambda\neq1$, $q\neq0$, $r\neq0$ and the conditions \eqref{eq:param_clas} or~\eqref{eq:param_mir} are satisfied in the classical or mirrored cases respectively. Then, for any differentiable and decreasing probability density function, we have
\begin{equation}\label{eq:stam-like}
\left(\Xi^{\frac{1}{p^*}}_{p^*,(2-r)\theta,\theta}[f]\phi_{q,r}^{\frac{qr}{\lambda-1}}[f]\right)^{\eta(p,\lambda)}\geqslant\widetilde \kappa_{p,\lambda}^{(0)},
\end{equation}
where
$$
\theta=\frac{q}{1-\lambda},
$$
and
$$
\eta(p,\lambda)=\begin{cases}
1,\quad &{\rm classical\ case,}\\
p^*-1,\quad& {\rm mirrored\ case}.
\end{cases}
$$
The optimal constant is given by
$$
\widetilde \kappa^{(0)}_{p,\lambda}=\begin{cases}
K_{p,\lambda}^{(0)},&\quad{\rm classical \ case,}\\[2mm]
\kappa_{p,\lambda}^{(0)},&\quad{\rm mirrored \ case.}
\end{cases}
$$
\end{theorem}
\begin{proof}
We start from the generalized moment-entropy inequalities in Eq.~\eqref{ineq:bip_E-M} and~\eqref{ineq:bip_E-M_mirrored}, which are satisfied under the conditions in Eqs.~\eqref{eq:param_clas} and~\eqref{eq:param_mir} in the classical and mirrored cases respectively. They can be written in a compact form using the notation introduced in the theorem,
\begin{equation}\label{eq:interm2bis}
\widetilde \kappa _{p,\lambda}^{(0)}\leqslant\left(\frac{\sigma_{p^*}[f]}{N_{\lambda}[f]}\right)^{\eta(p,\lambda)}
\end{equation}
We next apply \eqref{eq:interm2bis} to a biparametric down transformed density. Taking into account the identities Eqs. \eqref{eq:Ren_down} and \eqref{eq:gener_down_mom2}, we obtain
	\begin{equation}\label{eq:interm2}
	\widetilde \kappa _{p,\lambda}^{(0)}\leqslant\left(\frac{\sigma_{p^*}[f^{\Downarrow}_{\alpha,\gamma}]}{N_{\lambda}[f^{\Downarrow}_{\alpha,\gamma}]}\right)^{\eta(p,\lambda)}
	=\left(\frac{\Xi^{\frac{1}{p^*}}_{p^*,\alpha,\gamma}[f]}
	{\phi^{2\gamma-\alpha}_{(1-\lambda)\gamma,2-\alpha/\gamma}[f]}\right)^{\eta(p,\lambda)}.
	\end{equation}
	under the conditions \eqref{eq:cond_cumul}, $p^*\geqslant0$, $\gamma\neq0$ and $\alpha\neq2\gamma$. Adopting the notation
	$$
	q=(1-\lambda)\gamma, \quad r=2-\frac{\alpha}{\gamma},
	$$
	we deduce after straightforward manipulations that
	$$
	\gamma=\frac{q}{1-\lambda}, \quad \alpha=\frac{(2-r)q}{1-\lambda}, \quad 2\gamma-\alpha=\frac{rq}{1-\lambda},
	$$
	and Eq. \eqref{eq:stam-like} is established by inserting the previous equalities into Eq. \eqref{eq:interm2}. Notice that the conditions $\gamma\neq0$ and $\alpha\neq 2\gamma$ change into $q\neq0$, $r\neq0$, completing the proof in the classical case.
\end{proof}

\begin{remark}\label{rem:connect}
An interesting particular case is deduced by letting
\begin{equation}\label{eq:part_case}
\theta=\frac{q}{1-\lambda}=1.
\end{equation}
In this case, using Eq.~\eqref{eq:gdm1}, the inequality Eq. \eqref{eq:stam-like} becomes (for densities such that $\lim\limits_{x\to x_f}f(x)=0$)
$$
\widetilde \kappa_{p,\lambda}^{(0)}\leqslant\left(\frac{\Xi^{\frac{1}{p^*}}_{p^*,2-r,1}[f]}{\phi_{q,r}^{r}[f]}\right)^{\eta(p,\lambda)}
=|r|^{-\eta(p,\lambda)}\left(\frac{1}{N_{1+p^*r}[f]\phi_{q,r}[f]}\right)^{r\eta(p,\lambda)},
$$
or equivalently,
\begin{equation}\label{eq:interm3}
\big(N_{1+p^*r}[f]\phi_{q,r}[f]\big)^{r\eta(p,\lambda)}\leqslant\frac{|r|^{-\eta(p,\lambda)}}{\widetilde \kappa_{p,\lambda}^{(0)}}.
\end{equation}
Noticing that the conditions \eqref{eq:param_clas} and~\eqref{eq:param_mir} imply $\lambda>0$ and respectively $\lambda<0$ in the classical and mirrored cases, we infer from \eqref{eq:part_case} that
$$
q=1-\lambda\begin{cases} <1,& {\rm classical \ case,} \\ >1, &{\rm mirrored \ case.}\end{cases}
$$
It is easy to prove that Eq. \eqref{eq:interm3} with $\eta(p,\lambda)=1$ (classical case) transforms (by easy manipulations) into the triparametric Stam inequality in the mirrored domain \eqref{ineq:trip_Stam_extended} established in \cite[Theorem 5.1]{IP2025}, as expected. Viceversa, letting $\eta(p,\lambda)=p^*-1$ in Eq. \eqref{eq:interm3} (corresponding to the mirrored domain), we readily arrive to the classical domain of the triparametric Stam inequality by a reparametrization. This fact is remarkable since it shows that the so-called classical (moment-entropy inequality) and mirrored (triparametric Stam inequality) domains are connected in a continuous form, and the same is true for the opposite combination of domains.
\end{remark}
In the limiting case $\lambda=1$, we also derive a simple informational inequality relating the down-moments and the Shannon entropy.
\begin{theorem}\label{th:gendown_Shannon}
Let $p\in\Rset$ be such that $p^*>0$ and $(\alpha,\beta)\in\Rset^2$. Then
\begin{equation}\label{ineq:gendown_Shannon}
\Xi_{p^*,\alpha,\beta}^{\frac{1}{p^*}}[f]\geqslant K^{(0)}_{p,1}e^{\alpha S[f]}e^{\beta\left\langle\log|f'|\right\rangle_f},
\end{equation}
for any differentiable and decreasing probability density function.
\end{theorem}
\begin{proof}
We start from the  moment-entropy inequality \eqref{ineq:bip_E-M} and we apply it for $\lambda=1$ (which restricts the application to the classical domain of parameters) and to a biparametric down transformed density. We thus deduce that
\begin{equation}\label{eq:interm24}
\frac{\sigma_{p^*}[f^{\Downarrow}_{\alpha,\beta}]}{e^{S[f^{\Downarrow}_{\alpha,\beta}]}}\geqslant K^{(0)}_{p,1}.
\end{equation}
We next employ the identities \eqref{eq:gener_down_mom} and \eqref{eq:Shan_down} in the left-hand side of \eqref{eq:interm24} to find
$$
\frac{\Xi_{p^*,\alpha,\beta}^{\frac{1}{p^*}}[f]}{e^{\alpha S[f]}e^{\beta\left\langle\log|f'|\right\rangle_f}}\geqslant K^{(0)}_{p^*,1},
$$
which is equivalent to \eqref{ineq:gendown_Shannon}.
\end{proof}

\medskip

\noindent \textbf{Remark:} Let us denote by $\rho_{p,\beta,\lambda}$ the minimizers of the cumulative moment-entropy inequalities~\eqref{ineq:E-M-cumul} and~\eqref{ineq:cum_ent_mir} in both classical and mirrored domains. The minimizers of the inequalities in Theorems~\ref{th:stam-like} and~\ref{th:gendown_Shannon} are then given by
$$
f_{\rm min}=\mathfrak U_{\alpha}[\rho_{p,\beta,\lambda}], \quad{\rm and\ respectively,}\quad f_{\rm min}=\mathfrak U_{\alpha}[\rho_{p,\beta,1}].
$$

\subsection{A generalized cumulative moment-entropy inequality}

The next informational inequality is obtained by raising up one level the triparametric Stam inequality by applying it to a biparametric up transformed density. To this end, we introduce the notation
\begin{equation}\label{eq:gener_ent}
\widetilde{N}_{p,q,\alpha,\beta}[f]:=\phi_{p,q}[f^{\Uparrow}_{\alpha,\beta}].
\end{equation}
In order to express $\widetilde{N}_{p,q,\alpha,\beta}$ in an integral form, using Definition~\ref{def:up} and the fact that
$$
\frac{d \aup(u)}{du}=-\frac{[(\alpha-2)x(u)]^{\frac\alpha{2-\alpha}}}{f(x(u))}
$$
established in \cite[Section 2.2]{IP2025}, we calculate first, for $\alpha\neq2$,
\begin{equation}\label{eq:interm4}
\begin{split}
\phi_{p,q}^{pq}[\aup]&=\int_{\Rset}\aup(u)^{1+p(q-2)}|(\aup)'(u)|^p\,du\\
&=\int_{\Rset}|(\alpha-2)x(u)|^{\frac{p(q-2)}{2-\alpha}}|(\alpha-2)x(u)|^{\frac{p\alpha}{2-\alpha}}f(x)^{1-p}\,dx\\
&=\int_{\Rset}|(\alpha-2)x|^{\frac{p(q+\alpha-2)}{2-\alpha}}f^{1-p}(x)\,dx.
\end{split}
\end{equation}
Observe that, if we let $\alpha=2-q$ in the previous calculation, we find the entropy $N_{1-p}[f]^{1/q}$ at the end, as already stated in \cite[Lemma 3.1]{IP2025}. By combining now the definition of the biparametric up transformation with \eqref{eq:interm4}, we can further calculate (provided $\alpha\neq 2\beta$)
\begin{equation}\label{eq:interm5}
\begin{split}
\widetilde{N}_{p,q,\alpha,\beta}^{pq}[f]&=\phi_{p,q}^{pq}[\mathfrak{U}_{\frac{\alpha}{\beta}}[\mathfrak{E}_{\beta^{-1}}[f]]]\\
&=\int_{\Rset}\left|\left(\frac{\alpha}{\beta}-2\right)y\right|^{\frac{p(q+\alpha/\beta-2)}{2-\alpha/\beta}}\mathfrak{E}_{\frac{1}{\beta}}^{1-p}[f](y)\,dy\\
&=\int_{\Rset}\left|\left(\frac{\alpha}{\beta}-2\right)\int_{0}^xf^{\frac{\beta-1}{\beta}}(t)\,dt\right|^{\frac{p(q+\alpha/\beta-2)}{2-\alpha/\beta}}f^{\frac{\beta-p}{\beta}}(x)\,dx.
\end{split}
\end{equation}
As an interesting particular case of Eq. \eqref{eq:interm5}, letting $\alpha=(2-q)\beta$ in the previous identity, we get
\begin{equation}\label{eq:part_case2}
\widetilde{N}_{p,q,(2-q)\beta,\beta}[f]=\left[\int_{\Rset}f^{\frac{\beta-p}{\beta}}(x)\,dx\right]^{\frac{1}{pq}}=N_{\frac{\beta-p}{\beta}}^{\frac{1}{q\beta}}[f].
\end{equation}
We can state and prove a rather general informational inequality involving the cumulative moments and the functionals introduced in \eqref{eq:gener_ent}.
\begin{theorem}\label{th:mir_cumul}
Let $(p,q,\alpha,\beta,\lambda)\in\Rset^5$ be such that either $p\geqslant1$ and \eqref{eq:sign_cond1}, or $p<1$ and \eqref{eq:sign_cond2} are satisfied. Assume also that $\beta\neq0$, $\lambda\neq1$ and $\alpha\neq2\beta$. For any probability density function such that $f^{\Uparrow}_{\alpha,\beta}$ is absolutely continuous on its support, we have the following inequality
\begin{equation}\label{ineq:general}
\left(\mu_{\frac{\beta(\lambda-1)}{2\beta-\alpha},\frac{1}{\beta}}^{\frac{1}{1-\lambda}}[f]\widetilde{N}_{p,q,\alpha,\beta}[f]\right)^{\theta(q,\lambda)}
\geqslant\widetilde{K}(p,q,\lambda,\alpha,\beta),
\end{equation}
where $\theta(q,\lambda)$ is the exponent in the triparametric Stam inequality \eqref{ineq:trip_Stam_extended}. The minimizers of \eqref{ineq:general} are given by
\begin{equation}\label{eq:min_general}
f_{\rm min}:=\mathfrak{D}_{\alpha,\beta}[\widetilde{g}_{p,q,\lambda}],
\end{equation}
where $\widetilde{g}_{p,q,\lambda}$ are the minimizers of the triparametric Stam inequality \eqref{ineq:trip_Stam_extended}.
\end{theorem}
\begin{proof}
The proof follows readily by applying the triparametric Stam inequality \eqref{ineq:trip_Stam_extended} to the biparametric up transformed density $f^{\Uparrow}_{\alpha,\beta}$. Recalling \eqref{eq:Ren_up} and the definition \eqref{eq:gener_ent}, we directly deduce the informational inequality \eqref{ineq:general}, where the optimal constant is given by
$$
\widetilde{K}(p,q,\lambda,\alpha,\beta):=\left|\frac{\beta}{\alpha-2\beta}\right|^{\frac{\beta\theta(q,\lambda)}{\alpha-2\beta}}\kappa^{(1)}_{p,q,\lambda},
$$
with $\theta(q,\lambda)$ and $\kappa^{(1)}_{p,q,\lambda}$ being the exponent and optimal constant in the triparametric Stam inequality \eqref{ineq:trip_Stam_extended}. Since the triparametric Stam inequality is saturated by $\widetilde{g}_{p,q,\lambda}$, the minimizers of the inequality \eqref{ineq:general} satisfy
$$
\mathfrak{U}_{\alpha,\beta}[f_{\rm min}]=\widetilde{g}_{p,q,\lambda},
$$
which leads to \eqref{eq:min_general}.
\end{proof}

\noindent \textbf{Particular cases.} We give below two interesting particular cases of \eqref{ineq:general}.

\medskip

$\bullet$ \textbf{Case $\alpha=(2-q)\beta$.} Taking into account \eqref{eq:part_case2}, the inequality \eqref{ineq:general} reduces to a cumulative moment-entropy inequality, more precisely,
\begin{equation}\label{eq:cum_E-M}
\left(\mu_{\frac{\lambda-1}{q},\frac{1}{\beta}}^{\frac{1}{1-\lambda}}[f]N_{\frac{\beta-p}{\beta}}^{\frac{1}{q\beta}}[f]\right)^{\theta(q,\lambda)}
=\left(\frac{N_{\frac{\beta-p}{\beta}}[f]}{\sigma_{\frac{\lambda-1}{q},\frac{1}{\beta}}[f]}\right)^{\frac{\theta(q,\lambda)}{q\beta}}
\geqslant \widetilde{K}(p,q,\lambda,(2-q)\beta,\beta).
\end{equation}
Note that the inequality \eqref{eq:cum_E-M} can be seen as an extended cumulative moment-entropy inequality, since it generalizes \eqref{ineq:E-M-cumul} by enhancing the domain of the parameters in which it applies.

$\bullet$ \textbf{Case $\beta=1$.} In this case, the inequality Eq. \eqref{ineq:general} reduces to the simpler form (taking into account that the cumulative integrals are canceled both in the definition of the cumulative moment and in \eqref{eq:interm5})
$$
\left(\mu_{\frac{\lambda-1}{2-\alpha}}^{\frac{1}{1-\lambda}}[f]\int_{\Rset}f(x)^{1-p}|(\alpha-2)x|^{\frac{p(q+\alpha-2)}{2-\alpha}}\,dx\right)^{\theta(q,\lambda)}
\geqslant\left|\frac{1}{\alpha-2}\right|^{\frac{\theta(q,\lambda)}{\alpha-2}}\kappa^{(1)}_{p,q,\lambda}.
$$

\subsection{A cumulative upper-moment--moment inequality}

We derive in this section an informational inequality relating the cumulative upper moments introduced in \eqref{def:cumul_upper} to the standard cumulative moments defined in \eqref{def:cum_mom}.
\begin{theorem}\label{th:cum_upper}
Let $(\alpha,\beta)\in\Rset^2$ be such that $\alpha\neq2\beta$ and $\beta\neq0$.

\medskip

(a) If $(p,\lambda)\in\Rset^2$ are such that \eqref{eq:param_clas} is satisfied, we have
\begin{equation}\label{ineq:clas_upper}
M_{p^*,\alpha,\beta}^{\frac{1}{p^*}}[f]\sigma_{\frac{\beta(\lambda-1)}{2\beta-\alpha},\frac{1}{\beta}}^{\frac{1}{2\beta-\alpha}}[f]
\geqslant\left|\frac{2\beta-\alpha}{\beta}\right|^{\frac{\beta}{\alpha-2\beta}}K^{(0)}_{p,\lambda},
\end{equation}
for any probability density function $f$.

\medskip

(b) If $(p,\lambda)\in\Rset^2$ are such that \eqref{eq:param_mir} is satisfied, we have
\begin{equation}\label{ineq:mir_upper}
M_{p^*,\alpha,\beta}^{\frac{1}{p}}[f]\sigma_{\frac{\beta(\lambda-1)}{2\beta-\alpha},\frac{1}{\beta}}^{\frac{p^*-1}{2\beta-\alpha}}[f]
\geqslant\left|\frac{2\beta-\alpha}{\beta}\right|^{\frac{\beta(p^*-1)}{\alpha-2\beta}}\kappa^{(0)}_{p,\lambda}
\end{equation}
for any probability density function $f$.
\end{theorem}
\begin{proof}
We start from the classical moment-entropy inequality \eqref{ineq:bip_E-M}, that we apply to a biparametric up transformed density, obtaining thus
\begin{equation}\label{eq:interm25}
\frac{\sigma_{p^*}[f^{\Uparrow}_{\alpha,\beta}]}{N_{\lambda}[f^{\Uparrow}_{\alpha,\beta}]}\geqslant K^{(0)}_{p,\lambda}.
\end{equation}
Since $\alpha\neq2\beta$ and $\beta\neq0$, we next infer from \eqref{eq:Ren_up} that
$$
N_{\lambda}[f^{\Uparrow}_{\alpha,\beta}]=\left|\frac{2\beta-\alpha}{\beta}\right|^{\frac{\beta}{\alpha-2\beta}}
\mu_{\frac{\beta(\lambda-1)}{2\beta-\alpha},\frac{1}{\beta}}^{\frac{1}{1-\lambda}}[f]
=\left|\frac{2\beta-\alpha}{\beta}\right|^{\frac{\beta}{\alpha-2\beta}}
\sigma_{\frac{\beta(\lambda-1)}{2\beta-\alpha},\frac{1}{\beta}}^{\frac{1}{\alpha-2\beta}}[f].
$$
Inserting the above identity and the equality in \eqref{eq:cum_upper} into \eqref{eq:interm25}, we readily obtain the inequality \eqref{ineq:clas_upper}. In order to prove the inequality \eqref{ineq:mir_upper}, we proceed in exactly the same way as above, but starting from the moment-entropy inequality in the mirrored domain \eqref{ineq:bip_E-M_mirrored}, which features a general exponent $p^*-1$, completing the proof.
\end{proof}

\medskip

\noindent \textbf{Calculation of the minimizer in the classical domain.} The minimizer of the inequality \eqref{ineq:clas_upper} is given by
$$
f_{\rm min}:=\mathfrak{D}_{\alpha,\beta}[g_{p,\lambda}],
$$
while the minimizer of the inequality \eqref{ineq:mir_upper} is given by
$$
\overline{f}_{\rm min}:=\mathfrak{D}_{\alpha,\beta}[\overline{g}_{p,\lambda}]=\mathfrak{D}_{\alpha,\beta}[g_{1-\lambda,1-p}],
$$
in the notation introduced in \eqref{eq:not_min}. We next calculate in detail $f_{\rm min}$, in order to show that it is expressed in terms of generalized trigonometric functions. The calculation is rather tedious, but it can be also seen as an example of application of the biparametric down transform to an explicit probability density function, while demonstrating the convenience of using generalized trigonometric functions due to the significant simplifications occuring in the calculation process. It follows from \eqref{eq:escort_down} and \cite[Proposition 4.1]{IP2025} that
\begin{equation*}
\begin{split}
\mathfrak{D}_{\alpha,\beta}[g_{p,\lambda}]&=\mathfrak{E}_{\beta}\mathfrak{D}_{\frac{\alpha}{\beta}}[g_{p,\lambda}]\\
&=\mathfrak{E}_{\beta}\left[C_{p,\lambda}\left|\left(\frac{\alpha}{\beta}-2\right)s\right|^{\frac{\alpha+(\lambda-2)\beta}{2\beta-\alpha}}
\left|\left[\left(\frac{\alpha}{\beta}-2\right)s\right]^{\frac{\beta(\lambda-1)}{2\beta-\alpha}}-a_{p,\lambda}^{\lambda-1}\right|^{-\frac{1}{p}}\right]\\
&=\varrho(s(u))^{\beta},
\end{split}
\end{equation*}
where
\begin{equation}\label{eq:minimizer}
\varrho(s):=C_{p,\lambda}\left|\left(\frac{\alpha}{\beta}-2\right)s\right|^{\frac{\alpha+(\lambda-2)\beta}{2\beta-\alpha}}
\left|\left[\left(\frac{\alpha}{\beta}-2\right)s\right]^{\frac{\beta(\lambda-1)}{2\beta-\alpha}}-a_{p,\lambda}^{\lambda-1}\right|^{-\frac{1}{p}},
\end{equation}
with
$$
C_{p,\lambda}:=\frac{|1-\lambda|^{\frac{1}{p}}}{p^*a_{p,\lambda}^{\frac{\lambda-1}{p^*}}}.
$$
Taking into account the change of variable in the differential-escort transformation, we deduce that
\begin{equation*}
\begin{split}
u(s)&=\int_{s_0}^s\varrho(s^*)^{1-\beta}\,ds^*\\
&=C_{p,\lambda}^{1-\beta}\int_{s_0}^s\left[\left(\frac{\alpha}{\beta}-2\right)s^*\right]^{\frac{(1-\beta)(\alpha+(\lambda-2)\beta)}{2\beta-\alpha}}
\left|\left[\left(\frac{\alpha}{\beta}-2\right)s^*\right]^{\frac{\beta(\lambda-1)}{2\beta-\alpha}}-a_{p,\lambda}^{\lambda-1}\right|^{\frac{\beta-1}{p}}\,ds^*.
\end{split}
\end{equation*}
We introduce next the change of variable
\begin{equation}\label{eq:change1}
t^*:=\left|\frac{\alpha}{\beta}-2\right|^z\frac{|s^*|^{z}s^*}{z+1}, \quad z:=\frac{(\alpha+(\lambda-2)\beta)(1-\beta)}{2\beta-\alpha}
\end{equation}
and continue the calculation as follows:
\begin{equation*}
\begin{split}
u(s)&=C_{p,\lambda}^{1-\beta}\int_{t_0}^t\left|\left|(z+1)\left(\frac{\alpha}{\beta}-2\right)t^*\right|^{\frac{\beta(\lambda-1)}{(2\beta-\alpha)(z+1)}}
-a_{p,\lambda}^{\lambda-1}\right|^{\frac{\beta-1}{p}}\,dt^*\\
&=C_{p,\lambda}^{1-\beta}a_{p,\lambda}^{\frac{(\lambda-1)(\beta-1)}{p}}\int_{t_0}^t
\left|\frac{\left|(z+1)\left(\frac{\alpha}{\beta}-2\right)t^*\right|^{\frac{\beta(\lambda-1)}{(2\beta-\alpha)(z+1)}}}{a_{p,\lambda}^{\lambda-1}}-1\right|^{\frac{\beta-1}{p}}\,dt^*\\
&=C_{p,\lambda}^{1-\beta}a_{p,\lambda}^{\frac{(\lambda-1)(\beta-1)}{p}}\frac{1}{K}\int_{w_0}^{w}
\left|(w^*)^{\frac{\beta(\lambda-1)}{(2\beta-\alpha)(z+1)}}-1\right|^{\frac{\beta-1}{p}}\,dw^*,
\end{split}
\end{equation*}
where we have introduced a new change of variable
\begin{equation}\label{eq:change2}
w^*:=Kt^*, \quad K:=\frac{(z+1)(\alpha-2\beta)}{\beta a_{p,\lambda}^{\frac{(2\beta-\alpha)(z+1)}{\beta}}}.
\end{equation}
Observing that the definition of $z$ gives
$$
\frac{\beta(\lambda-1)}{(z+1)(2\beta-\alpha)}=\frac{\lambda-1}{2\beta-\alpha+\lambda(1-\beta)},
$$
and introducing the shorter notation
$$
\Psi(\alpha,\beta,\lambda):=2\beta-\alpha+\lambda(1-\beta),
$$
we finally deduce that
\begin{equation}\label{eq:interm26}
u(s)=\mathcal{K}\left[\arcsin_{\frac{p}{1-\beta},\frac{\lambda-1}{\Psi(\alpha,\beta,\lambda)}}(w)-C_0\right],
\end{equation}
where $C_0$ is the generalized arcsine function evaluated at $w_0$ and
$$
\mathcal{K}:=C_{p,\lambda}^{1-\beta}a_{p,\lambda}^{\frac{(\lambda-1)(\beta-1)}{p}}\frac{1}{K},
$$
the constant $K$ being defined in \eqref{eq:change2}. By undoing the changes of variables \eqref{eq:change2} and \eqref{eq:change1}, we find
$$
w^*=\frac{1}{a_{p,\lambda}^{\Psi(\alpha,\beta,\lambda)}}\left[\left(\frac{\alpha}{\beta}-2\right)s^*\right]^{z+1}=\widetilde{K}(s^*)^{z+1}
$$
and thus finally obtain from \eqref{eq:interm26} that
\begin{equation}\label{eq:interm27}
u(s)=\mathcal{K}\left[\arcsin_{\frac{p}{1-\beta},\frac{\lambda-1}{\Psi(\alpha,\beta,\lambda)}}(\widetilde{K}s^{z+1})-C_0\right].
\end{equation}
We deduce from \eqref{eq:interm27} that
\begin{equation}\label{eq:interm28}
\begin{split}
\sin_{\frac{p}{1-\beta},\frac{\lambda-1}{\Psi(\alpha,\beta,\lambda)}}&\left(\frac{u}{\mathcal{K}}+C_0\right)=\widetilde{K}s(u)^{z+1}
=\frac{\left[\left(\frac{\alpha}{\beta}-2\right)s(u)\right]^{z+1}}{a_{p,\lambda}^{\Psi(\alpha,\beta,\lambda)}}\\
&=\left[\frac{\left[\left(\frac{\alpha}{\beta}-2\right)s(u)\right]^{\frac{\beta}{2\beta-\alpha}}}{a_{p,\lambda}}\right]^{\Psi(\alpha,\beta,\lambda)}.
\end{split}
\end{equation}
Coming back to the initial calculation and employing the properties of the generalized trigonometric functions, we derive from \eqref{eq:minimizer} and \eqref{eq:interm28} that
\begin{equation*}
\begin{split}
\varrho(s)&=C_{p,\lambda}\left[\left(\frac{\alpha}{\beta}-2\right)s\right]^{\frac{\alpha+(\lambda-2)\beta}{2\beta-\alpha}}a_{p,\lambda}^{\frac{1-\lambda}{p}}
\left|\frac{\left[\left(\frac{\alpha}{\beta}-2\right)s\right]^{\frac{\beta(\lambda-1)}{2\beta-\alpha}}}{a_{p,\lambda}^{\lambda-1}}-1\right|^{-\frac{1}{p}}\\
&=C_{p,\lambda}\left[\left(\frac{\alpha}{\beta}-2\right)s\right]^{\frac{\alpha+(\lambda-2)\beta}{2\beta-\alpha}}a_{p,\lambda}^{\frac{1-\lambda}{p}}
\left|\left[\frac{\left[\left(\frac{\alpha}{\beta}-2\right)s\right]^{\frac{\beta}{2\beta-\alpha}}}{a_{p,\lambda}}\right]^{\lambda-1}-1\right|^{-\frac{1}{p}}\\
&=C_{p,\lambda}\left[\left(\frac{\alpha}{\beta}-2\right)s\right]^{\frac{\alpha+(\lambda-2)\beta}{2\beta-\alpha}}a_{p,\lambda}^{\frac{1-\lambda}{p}}
\left|\sin_{\frac{p}{1-\beta},\frac{\lambda-1}{\Psi(\alpha,\beta,\lambda)}}^{\frac{\lambda-1}{\Psi(\alpha,\beta,\lambda)}}\left(\frac{u}{\mathcal{K}}+C_0\right)-1\right|^{-\frac{1}{p}}\\
&=C_{p,\lambda}a_{p,\lambda}^{\frac{1-\lambda}{p}+\frac{\alpha+(\lambda-2)\beta}{\beta}}
\sin_{\frac{p}{1-\beta},\frac{\lambda-1}{\Psi(\alpha,\beta,\lambda)}}^{\frac{\alpha+(\lambda-2)\beta}{\beta\Psi(\alpha,\beta,\lambda)}}
\left(\frac{u}{\mathcal{K}}+C_0\right)
\cos_{\frac{p}{1-\beta},\frac{\lambda-1}{\Psi(\alpha,\beta,\lambda)}}^{\frac{1}{\beta-1}}\left(\frac{u}{\mathcal{K}}+C_0\right).
\end{split}
\end{equation*}
We thus conclude that the minimizer to the inequality \eqref{ineq:clas_upper} is expressed in terms of the generalized trigonometric functions as follows:
$$
\mathfrak{D}_{\alpha,\beta}[g_{p,\lambda}]=C_{p,\lambda}^{\beta}a_{p,\lambda}^{\alpha+(\lambda-2)\beta+\frac{(1-\lambda)\beta}{p}}
\sin_{\frac{p}{1-\beta},\frac{\lambda-1}{2\beta-\alpha+\lambda(1-\beta)}}^{\frac{\alpha+(\lambda-2)\beta}{2\beta-\alpha+\lambda(1-\beta)}}
\left(\frac{u}{\mathcal{K}}+C_0\right)
\cos_{\frac{p}{1-\beta},\frac{\lambda-1}{2\beta-\alpha+\lambda(1-\beta)}}^{\frac{\beta}{\beta-1}}\left(\frac{u}{\mathcal{K}}+C_0\right).
$$

\section{Monotonicity properties}\label{sec:monot}

A very important aspect when dealing with complex systems is to introduce ways of measuring their complexity. As discussed in \cite{LopezRuiz1995, LopezRuiz2005}, measuring the complexity of a system is a difficult task and criteria for complexity can vary, leading to different approaches to measuring it. In the last decades a number of complexity measures have been proposed, such as, for example, the Fisher-Shannon complexity \cite{Rudnicki2016} or the LMC-Rényi complexity \cite{Puertas2019}. Monotonicity properties of such complexity measures with respect to suitable transformations are thus important properties of these measures.

The aim of this section is, thus, to construct a number of complexity measures based on the moments, entropies and Fisher information of the probability density functions under consideration (both the classical ones and the ones introduced in this work) and establish monotonicity properties of them, with the help of the biparametric transformations introduced in Section \ref{sec:transf}. Let us thus define the following measures:

$\bullet$ the LMC-Rényi-complexity: for $0<p<q$,
\begin{equation}\label{def:EC}
C_{p,q}^{(N)}[f]:=\frac{N_{p}[f]}{N_q[f]}\geqslant1.
\end{equation}

$\bullet$ the moment-complexity: again for $p>q>0$,
\begin{equation}\label{def:MC}
C_{p,q}^{(\sigma)}[f]:=\frac{\sigma_p[f]}{\sigma_q[f]}\geqslant1.
\end{equation}

$\bullet$ the Fisher-complexity: for $p>q$ and $\lambda>0$,
\begin{equation}\label{def:FC}
C_{p,q,\lambda}^{(\phi)}[f]:=\frac{\phi_{p,\lambda}[f]}{\phi_{q,\lambda}[f]}\geqslant1.
\end{equation}

The first complexity measure, defined in \eqref{def:EC}, has been considered in \cite{Puertas2019}, where the following monotonicity property has been established, with respect to the differential-escort transformation: given two real numbers $p$, $q$ such that $p<q$, it follows from \cite[Property 10]{Puertas2019} that, for any $\alpha$, $\alpha'\in\Rset$ such that either $\alpha'>\alpha>0$ or $\alpha'<\alpha<0$, and for any probability density function $f$, we have
\begin{equation}\label{eq:monot_ent}
C_{p,q}[\mathfrak{E}_{\alpha'}[f]]\geqslant C_{p,q}[\mathfrak{E}_{\alpha}[f]].
\end{equation}
In the following lines, we exploit the monotonicity property \eqref{eq:monot_ent} and the previously established properties of the biparametric up and down transformations in order to derive monotonicity properties for the moment-complexity and the Fisher-complexity. For the first one of them, we have:
\begin{theorem}\label{prop:monot_moment}
Let $(p,q,\alpha,\gamma)\in\Rset^4$ be such that $p>q>0$, $\alpha\geqslant1$ and $\gamma\neq2$. Then
\begin{equation}\label{eq:monot_moment}
C_{p,q}^{(\sigma)}[\mathfrak{E}_{\alpha,\gamma}^{(\sigma)}[f]]\geqslant C_{p,q}^{(\sigma)}[f], \quad {\rm where} \quad \mathfrak{E}_{\alpha,\gamma}^{(\sigma)}:=\mathfrak{D}_{\gamma}\mathfrak{E}_{\alpha}\mathfrak{U}_{\gamma},
\end{equation}
for any probability density function $f$.
\end{theorem}
\begin{proof}
We infer from \eqref{eq:ME} that, for any $\gamma\in\Rset\setminus\{2\}$, we have
\begin{equation*}
C_{p,q}^{(\sigma)}[\mathfrak{D}_{\gamma}[f]]=\frac{\sigma_p[\mathfrak{D}_{\gamma}[f]]}{\sigma_q[\mathfrak{D}_{\gamma}[f]]}
=\left[\frac{N_{1+(2-\gamma)p}[f]}{N_{1+(2-\gamma)q}[f]}\right]^{\gamma-2}.
\end{equation*}
Since $p>q>0$, let us observe that

$\bullet$ if $\gamma>2$, then $1+(2-\gamma)p<1+(2-\gamma)q$, while $\gamma-2>0$. By applying the LMC-Rényi monotonicity property \eqref{eq:monot_ent}, we find that
\begin{equation*}
\begin{split}
C_{p,q}^{(\sigma)}[\mathfrak{D}_{\gamma}\mathfrak{E}_{\alpha}\mathfrak{U}_{\gamma}[f]]
&=C_{\lambda,\beta}^{(N)}[\mathfrak{E}_{\alpha}\mathfrak{U}_{\gamma}[f]]^{\gamma-2}\geqslant C_{\lambda,\beta}^{(N)}[\mathfrak{U}_{\gamma}[f]]^{\gamma-2}\\
&=C_{p,q}^{(\sigma)}[f],
\end{split}
\end{equation*}
where we have adopted the notation $\lambda:=1+(2-\gamma)p<1+(2-\gamma)q=:\beta$. The property \eqref{eq:monot_moment} is thus established in the case $\gamma>2$.

$\bullet$ if $\gamma<2$, then $\lambda=1+(2-\gamma)p>1+(2-\gamma)q=\beta$, in the previous notation, but now $\gamma-2>0$. The monotonicity property thus follows as in the case $\gamma>2$, but taking into account that, in this case, the Rényi entropy powers $N_{\lambda}$ and $N_{\beta}$ are ordered in the opposite way, but the fact that the power $\gamma-2$ is now negative preserves the claimed monotonicity.
\end{proof}

\noindent \textbf{Remark.} By applying the composition properties \eqref{eq:compDE} and \eqref{eq:compDU}, we deduce that
$$
\mathfrak{D}_{\gamma}\mathfrak{E}_{\alpha}\mathfrak{U}_{\gamma}=\frac{1}{|\alpha|}\mathfrak{D}_{2+\alpha(\gamma-2)}\mathfrak{U}_{\gamma}
\simeq\mathfrak{G}_{\alpha(2-\gamma),2-\gamma},
$$
where the transformation $\mathfrak{G}$ is defined in Proposition \ref{prop:comp_UD}. The previous calculation gives, modulo a scaling change (which is negligible for the complexity measures, as they are scaling-invariant), an easier practical expression to calculate the composition of transformations $\mathfrak{D}_{\gamma}\mathfrak{E}_{\alpha}\mathfrak{U}_{\gamma}$.

\medskip

We proceed in a similar way towards a monotonicity property for the Fisher complexity, that is stated below.
\begin{theorem}\label{prop:monot_Fisher}
Let $(p,q,\alpha,\gamma)\in\Rset^4$ be such that $p>q$, $\alpha\geqslant1$ and $\gamma\in(0,2)$. Then
\begin{equation}\label{eq:monot_Fisher}
C_{p,q,2-\gamma}^{(\phi)}[\mathfrak{E}_{\alpha,\gamma}^{(\phi)}[f]]\geqslant C_{p,q,2-\gamma}^{(\phi)}[f], \quad {\rm where} \quad \mathfrak{E}_{\alpha,\gamma}^{(\phi)}:=\mathfrak{U}_{\gamma}\mathfrak{E}_{\alpha}\mathfrak{D}_{\gamma},
\end{equation}
for any derivable and decreasing probability density function $f$.
\end{theorem}
\begin{proof}
We infer from \eqref{eq:EF} that
$$
\frac{N_{\lambda}[\mathfrak{D}_{\gamma}[f]]}{N_{\beta}[\mathfrak{D}_{\gamma}[f]]}=\left[\frac{\phi_{p,2-\gamma}[f]}{\phi_{q,2-\gamma}[f]}\right]^{2-\gamma},
$$
where $\lambda=1-p$ and $\beta=1-q$. The order $p>q$ ensures that $\lambda<\beta$. We are thus in a position to apply again \eqref{eq:monot_ent} as follows:
\begin{equation*}
\begin{split}
\left[C_{p,q,2-\gamma}^{(\phi)}[\mathfrak{U}_{\gamma}\mathfrak{E}_{\alpha}\mathfrak{D}_{\gamma}[f]]\right]^{2-\gamma}&
=C_{\lambda,\beta}^{(N)}[\mathfrak{E}_{\alpha}\mathfrak{D}_{\gamma}[f]]\\
&\geqslant C_{\lambda,\beta}^{(N)}[\mathfrak{D}_{\gamma}[f]]=\left[C_{p,q,2-\gamma}^{(\phi)}[f]\right]^{2-\gamma},
\end{split}
\end{equation*}
and the fact that $\gamma<2$ entails the monotonicity property in \eqref{eq:monot_Fisher}, completing the proof.
\end{proof}

We next introduce several complexity measures based on the informational functionals introduced in the present paper. More precisely, for any $(p,q,a,b)\in\Rset^4$ such that $p>q>0$, we introduce the following measures:

$\bullet$ the upper-moment complexity measure
\begin{equation}\label{def:comp_upper}
C_{p,q,a,b}^{(M)}[f]:=\frac{M_{p,a,b}^{\frac{1}{p}}[f]}{M_{q,a,b}^{\frac{1}{q}}[f]}\geqslant1,
\end{equation}
the inequality in \eqref{def:comp_upper} following from \eqref{eq:order_upper}.

$\bullet$ the down-moment complexity measure
\begin{equation}\label{def:comp_down}
C_{p,q,a,b}^{(\Xi)}[f]:=\frac{\Xi_{p,a,b}^{\frac{1}{p}}[f]}{\Xi_{q,a,b}^{\frac{1}{q}}[f]}\geqslant1,
\end{equation}
the inequality in \eqref{def:comp_down} following from \eqref{eq:order_gener_dM}.

Employing the monotonicity properties established in Theorem \ref{prop:monot_moment} and the definitions of the upper-moment, respectively the down-moment, we can derive monotonicity properties for these two new complexity measures.
\begin{theorem}\label{prop:monot_upper_down}
(a) Let $(p,q,a,b,\alpha,\gamma)\in\Rset^6$ such that $p>q>0$, $\alpha\geqslant1$ and $\gamma\neq2$. Set
$$
\mathfrak{E}_{\alpha;a,b,\gamma}^{(M)}:=\mathfrak{D}_{a,b}\mathfrak{D}_{\gamma}\mathfrak{E}_{\alpha}\mathfrak{U}_{\gamma}\mathfrak{U}_{a,b}.
$$
Then we have
\begin{equation}\label{eq:monot_upper}
C_{p,q,a,b}^{(M)}[\mathfrak{E}_{\alpha;a,b,\gamma}^{(M)}[f]]\geqslant C_{p,q,a,b}^{(M)}[f],
\end{equation}
for any probability density function $f$ for which the transformation is well-defined.

\medskip

(b) In the same conditions as in (a), we set
$$
\mathfrak{E}_{\alpha;a,b,\gamma}^{(\Xi)}:=\mathfrak{U}_{a,b}\mathfrak{D}_{\gamma}\mathfrak{E}_{\alpha}\mathfrak{U}_{\gamma}\mathfrak{D}_{a,b}.
$$
Then, a similar monotonicity property holds for the down-moment complexity:
\begin{equation}\label{eq:monot_down}
C_{p,q,a,b}^{(\Xi)}[\mathfrak{E}_{\alpha;a,b,\gamma}^{(\Xi)}[f]]\geqslant C_{p,q,a,b}^{(\Xi)}[f],
\end{equation}
for any derivable and decreasing probability density function $f$.
\end{theorem}

\noindent \textbf{Remark.} We deduce from \eqref{eq:compEU} that
$$
\mathfrak{E}_{\alpha}\mathfrak{U}_{\gamma}=\mathfrak{U}_{\overline{\gamma}}, \quad \overline{\gamma}=2+\frac{\gamma-2}{\alpha},
$$
and, consequently, it follows from \eqref{eq:compDU} that
\begin{equation}\label{eq:interm40}
\mathfrak{D}_{\gamma}\mathfrak{E}_{\alpha}\mathfrak{U}_{\gamma}[\mathfrak{U}_{a,b}[f]](s)=|(2-\gamma)s|^{\frac{1}{\alpha}-1}
\mathfrak{U}_{a,b}[f]\left(\frac{\alpha}{|2-\gamma|}|(2-\gamma)s|^{\frac{1}{\alpha}}\right).
\end{equation}
We thus observe that, in order for the transformation $\mathfrak{E}_{\alpha;a,b,\gamma}^{(M)}$ to be well-defined, we have to impose the technical condition that the right-hand side of \eqref{eq:interm40} is a decreasing function of $s$. Note that, for $s>0$, this property is automatically satisfied, since $\alpha\geqslant1$ and thus $1/\alpha-1\leqslant0$.

On the contrary, since the application of $\mathfrak{U}_{a,b}$ does not require any condition, the transformation $\mathfrak{E}_{\alpha;a,b,\gamma}^{(\Xi)}$ is well-defined for any differentiable and decreasing probability density function $f$.

\begin{proof}
(a) We apply the monotonicity property \eqref{eq:monot_moment} to the biparametric up transformed density $\mathfrak{U}_{a,b}[f]$. On the one hand, we have
\begin{equation}\label{eq:interm31}
C_{p,q}^{(\sigma)}[f^{\Uparrow}_{a,b}]=\frac{\sigma_{p}[f^{\Uparrow}_{a,b}]}{\sigma_q[f^{\Uparrow}_{a,b}]}
=\frac{M_{p,a,b}^{\frac{1}{p}}[f]}{M_{q,a,b}^{\frac{1}{q}}[f]}=C_{p,q,a,b}^{(M)}[f],
\end{equation}
which also implies the identity
\begin{equation}\label{eq:interm30}
C_{p,q}^{(\sigma)}[f]=C_{p,q,a,b}^{(M)}[f^{\Downarrow}_{a,b}],
\end{equation}
whenever the density $f$ is decreasing and differentiable. On the other hand, an application of \eqref{eq:interm30} in the left-hand side of the inequality \eqref{eq:monot_moment} gives
\begin{equation}\label{eq:interm32}
\begin{split}
C_{p,q}^{(\sigma)}[\mathfrak{D}_{\gamma}\mathfrak{E}_{\alpha}\mathfrak{U}_{\gamma}[f^{\Uparrow}_{a,b}]]&=
C_{p,q}^{(\sigma)}[\mathfrak{D}_{\gamma}\mathfrak{E}_{\alpha}\mathfrak{U}_{\gamma}\mathfrak{U}_{a,b}[f]]\\
&=C_{p,q,a,b}^{(M)}[\mathfrak{D}_{a,b}\mathfrak{D}_{\gamma}\mathfrak{E}_{\alpha}\mathfrak{U}_{\gamma}\mathfrak{U}_{a,b}[f]].
\end{split}
\end{equation}
The inequality \eqref{eq:monot_upper} follows from the inequality \eqref{eq:monot_moment} and the identities \eqref{eq:interm31} and \eqref{eq:interm32}.

\medskip

(b) We proceed in a similar way as in part (a), by applying the inequality \eqref{eq:monot_moment} this time to a biparametric down transformed density $\mathfrak{D}_{a,b}[f]$, for any differentiable and decreasing probability density function $f$. On the one hand, we have
\begin{equation}\label{eq:interm34}
C_{p,q}^{(\sigma)}[f^{\Downarrow}_{a,b}]=\frac{\sigma_{p}[f^{\Downarrow}_{a,b}]}{\sigma_q[f^{\Downarrow}_{a,b}]}
=\frac{\Xi_{p,a,b}^{\frac{1}{p}}[f]}{\Xi_{q,a,b}^{\frac{1}{q}}[f]}=C_{p,q,a,b}^{(\Xi)}[f],
\end{equation}
which also implies the identity
\begin{equation}\label{eq:interm33}
C_{p,q}^{(\sigma)}[f]=C_{p,q,a,b}^{(\Xi)}[f^{\Uparrow}_{a,b}].
\end{equation}
On the other hand, we compute the left-hand side in \eqref{eq:monot_moment} and, with the aid of \eqref{eq:interm33}, we deduce that
\begin{equation}\label{eq:interm35}
\begin{split}
C_{p,q}^{(\sigma)}[\mathfrak{D}_{\gamma}\mathfrak{E}_{\alpha}\mathfrak{U}_{\gamma}[f^{\Downarrow}_{a,b}]]&=
C_{p,q}^{(\sigma)}[\mathfrak{D}_{\gamma}\mathfrak{E}_{\alpha}\mathfrak{U}_{\gamma}\mathfrak{D}_{a,b}[f]]\\
&=C_{p,q,a,b}^{(\Xi)}[\mathfrak{U}_{a,b}\mathfrak{D}_{\gamma}\mathfrak{E}_{\alpha}\mathfrak{U}_{\gamma}\mathfrak{D}_{a,b}[f]].
\end{split}
\end{equation}
The inequality \eqref{eq:monot_down} follows then from the inequality \eqref{eq:monot_moment} and the identities \eqref{eq:interm34} and \eqref{eq:interm35}.
\end{proof}
We finally introduce a complexity measure related to the down-Fisher measure and based on the inequality \eqref{eq:order_gener_dF}. For any $(p,q,s,\lambda)\in\Rset^4$ such that $p>q>0$ and $s\neq p$, we define
\begin{equation}\label{def:comp_dF}
C_{p,q,s,\lambda}^{(\varphi)}[f]:=\frac{\varphi_{p,s,\lambda}^{\frac{1}{p}}[f]}{\varphi_{q,\frac{qs}p,\lambda,}^{\frac{1}{q}}[f]}\geqslant1.
\end{equation}
We then have the following monotonicity property:
\begin{theorem}\label{prop:monot_dF}
Let
$$
\mathfrak{E}_{\alpha;p,s,\lambda}^{(\varphi)}:=\mathfrak{U}_{\frac{\lambda p}{p-s}}\mathfrak{U}_{\frac{p+s}{p}}\mathfrak{E}_{\alpha}\mathfrak{D}_{\frac{p+s}{p}}\mathfrak{D}_{\frac{\lambda p}{p-s}}.
$$
In the previous notation and conditions, for any $\alpha\geqslant1$ the following inequality
\begin{equation}\label{eq:monot_dF}
C_{p,q,s,\lambda}^{(\varphi)}[\mathfrak{E}_{\alpha;p,s,\lambda}^{(\varphi)}[f]]\geqslant C_{p,q,s,\lambda}^{(\varphi)}[f],
\end{equation}
holds true for any probability density function which is decreasing, differentiable up to second order and such that the condition \eqref{cond:down_twice} is satisfied for $(\alpha,\beta)=(p\lambda/(p-s),1)$.
\end{theorem}
\begin{proof}
We start from the monotonicity property of the Fisher complexity given in \eqref{eq:monot_Fisher}, which we apply to a down transformed density. We thus have
\begin{equation}\label{eq:interm36}
C_{p,q,2-\gamma}^{(\phi)}[\mathfrak{U}_{\gamma}\mathfrak{E}_{\alpha}\mathfrak{D}_{\gamma}[f^{\downarrow}_{a}]]\geqslant C_{p,q,2-\gamma}^{(\phi)}[f^{\downarrow}_{a}],
\end{equation}
for any $\gamma\in(0,2)$. We employ \eqref{eq:gener_down_Fisher} to calculate first the right-hand side of \eqref{eq:interm36} and we obtain
$$
C_{p,q,2-\gamma}^{(\phi)}[f^{\downarrow}_{a}]
=\left(\frac{\varphi_{p,p(\gamma-1),a(2-\gamma)}^{\frac{1}{p}}[f]}{\varphi_{q,q(\gamma-1),a(2-\gamma)}^{\frac{1}{q}}[f]}\right)^{\frac1{2-\gamma}}.
$$
Letting in the previous equality
$$
\lambda:=a(2-\gamma), \quad s:=p(\gamma-1), \quad {\rm that \ is,} \quad \gamma=\frac{p+s}{p}, \quad a=\frac{p\lambda}{p-s},
$$
we deduce that
\begin{equation}\label{eq:interm38}
C_{p,q,2-\gamma}^{(\phi)}[f^{\downarrow}_{a}]=C_{p,q,s,\lambda}^{(\varphi)}[f],
\end{equation}
which also entails
\begin{equation}\label{eq:interm37}
C_{p,q,2-\gamma}^{(\phi)}[f]=C_{p,q,s,\lambda}^{(\varphi)}[f^{\uparrow}_{a}].
\end{equation}
By applying \eqref{eq:interm37} and \eqref{eq:interm38} in \eqref{eq:interm36}, we deduce that
\begin{equation}\label{eq:interm39}
C_{p,q,s,\lambda}^{(\varphi)}[\mathfrak{E}_{\alpha;p,s,\lambda}^{(\varphi)}[f]]=C_{p,q,s,\lambda}^{(\varphi)}[\mathfrak{U}_{a}\mathfrak{U}_{\gamma}\mathfrak{E}_{\alpha}\mathfrak{D}_{\gamma}\mathfrak{D}_{a}[f]]\geqslant C_{p,q,s,\lambda}^{(\varphi)}[f],
\end{equation}
for any $f$ as in the statement of the theorem, noting that the condition \eqref{cond:down_twice} assumed in the statement allows for applying the down transformation twice.
The proof is thus complete.
\end{proof}

\noindent \textbf{Remark. Group structure.} We observe that, for $\alpha,\beta\in\Rset$ and in the notation introduced in Theorem \ref{prop:monot_upper_down}, we have
$$
\mathfrak{E}_{\alpha;a,b,\gamma}^{(M)}\mathfrak{E}_{\beta;a,b,\gamma}^{(M)}=\mathfrak{E}_{\alpha\beta;a,b,\gamma}^{(M)}
$$
and
$$
\mathfrak{E}_{\alpha;a,b,\gamma}^{(\Xi)}\mathfrak{E}_{\beta;a,b,\gamma}^{(\Xi)}=\mathfrak{E}_{\alpha\beta;a,b,\gamma}^{(\Xi)},
$$
which reveal a group structure formed by these transformations if we restrict ourselves to nonzero parameters $\alpha$ and $\beta$. The same is valid for the transformation $\mathfrak{E}_{\alpha;p,s,\lambda}^{(\varphi)}$ introduced in Theorem \ref{prop:monot_dF}, as well as for the transformations $\mathfrak{E}_{\alpha;\gamma}^{(\sigma)}$ and $\mathfrak{E}_{\alpha;\gamma}^{(\phi)}$ employed respectively in Theorems \ref{prop:monot_moment} and \ref{prop:monot_Fisher}. Let us stress here that these group structure are inherited from the analogous group structure of the differential-escort transformation, since all the previous composed transformations are obtained from the differential-escort transformation by algebraic conjugation with up and down transformations.

\section{Conclusions}\label{sec:conclusions}

Two new biparametric families of transformations, called biparametric up/down transformations, acting on suitable classes of probability density functions, have been introduced in this paper, generalizing the previously defined (one-parameter) up and down transformations, studied by the authors in their previous works. More precisely, the biparametric down transformation acts on decreasing and differentiable density functions, while the biparametric \DP{up\sout{down}} transformation can be applied to any density function, but producing a differentiable and decreasing density as result. Indeed, the biparametric families have as particular case the one-parameter up and down transformation if we set the second parameter to be equal to one. As expected, biparametric up/down transformations with the same parameters are mutually inverse.

The most remarkable property of these transformations is that they interpolate between the branches of up and down transformations with a single parameter, on the one hand, and the generalized differential-escort transformations, on the other hand. This interpolation of transformations leads to a number of interesting properties, among which the most significant one is, in our opinion, the fact that we may deduce informational inequalities connecting the classical and mirrored domains of parameters (that have been observed separately in previous works). To be more precise, the inequalities established by employing the biparametric up and down transformations interpolate between the classical moment-entropy inequality and the mirrored triparametric Stam inequality and viceversa, while the classical and mirrored domains of the same inequality remain disconnected. We give a visual description of this interpolation between inequalities and domains in Figure~\ref{figure}.

We believe that the introduction of the second parameter and the subsequent interpolation refines and enriches the structure identified in the study of the properties of the up and down transformations performed in previous works.

With the aid of the biparametric up and down transformations, we have defined new informational functionals, called \emph{down-moments} and \emph{cumulative upper-moments}, and establish sharp informational inequalities connecting them to previously defined functionals. Minimizers and optimal constants are also given, and in some cases the minimizers are expressed in terms of biparametric generalized trigonometric functions. One more remarkable fact related to these new functionals is that the down-moments allow to interpolate, for density functions with some particular properties, between the $p$-th absolute moment and the R\'enyi entropy power.

As a byproduct of the biparametric up and down transformations and of the informational functionals motivated by them, we have defined a number of different measures of statistical complexity and established monotonicity properties for them, with the help of transformations obtained by composing, in the form of algebraic conjugation, the up/down transformation with either one or two parameters and the differential-escort transformations. Moreover, we reveal a group structure on these families of composed transformations achieving the monotonicity properties. We believe that this intricate structure of transformations might be the starting point for further interesting theoretical research and applications.

\begin{figure}[H]
	\centering
	\begin{tikzpicture}
	\node  (down) at (5.8,-2.6)    {\huge$\mathfrak D_{\alpha,\beta}$};
	\node   at (0,-2.6)    {\huge$\mathfrak U_{\alpha,\beta}$};
	\node (EM) at (0,0)  {\Large $\frac{\sigma_{p^*,1+\beta-\lambda}[f]}{N_\lambda[f]}\,$\Large $\geqslant K_{p,\beta,\lambda}^{(0)}$};
	\node (EMmir) at (6,0)  {\Large $\left(\frac{\sigma_{p^*,\delta}[f]}{N_\lambda[f]}\right)^{\delta(p^*-1)}\,$\Large$\geqslant \widetilde\kappa_{p,\frac{\delta+\lambda-1}{\delta}}^{(0)}$};
	\node(Stam) at(0,-5) {\large$ \left(\phi_{p,q}[f]\,N_\lambda[f]\right)^{1+q-\lambda}\geqslant \kappa_{p,q,\lambda}^{(1)}$};
	\node(Stammir) at(6,-5) {\large$\left(\phi_{p,q}[f]\,N_\lambda[f]\right)^{-q}\geqslant \kappa_{p,q,\lambda}^{(1)}$};
	\draw[-{Stealth[length=5mm, open, round]}] (EM) to [bend left=25] (Stammir);
	\draw[-{Latex[length=5mm]}] (Stam) to [bend left=15] (EMmir);
	\draw[-{Stealth[length=5mm, open, round]}] (Stammir) to [bend left=25] (EM);
	\draw[-{Latex[length=5mm]}] (EMmir) to [bend left=15] (Stam);
	\end{tikzpicture}
	\caption{Graphical scheme of the connections between classical and mirrored domains.}\label{figure}
\end{figure}
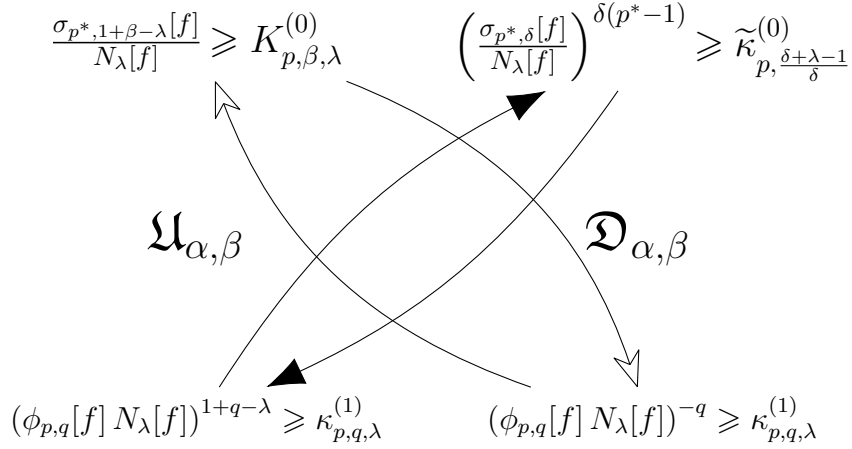

\appendix

\section{Appendix. Combined transformations}

We gather in this final and technical section the result of the compositions between up/down and differential-escort transformations. Some of these combined transformations arose in a natural way throughout Section \ref{sec:monot}. For the sake of completeness, we also give the outcome of some compositions of transformations that have not been employed in the main part of this paper. We begin with the following result regarding the composition between down or up transformations and the differential-escort transformation.
\begin{proposition}\label{prop:comp_UDE}
Let $(\alpha,\gamma)\in\Rset^2$ such that $\alpha\neq0$ and let $f$ be any probability density function. Assuming that $f$ is decreasing on its support, we have
\begin{equation}\label{eq:compDE}
\mathfrak{D}_{\gamma}\mathfrak{E}_{\alpha}[f]=\frac{1}{|\alpha|}\mathfrak{D}_{2+\alpha(\gamma-2)}[f].
\end{equation}
If we assume furthermore that $\gamma\neq2$, for any probability density function $f$ we have
\begin{equation}\label{eq:compEU}
\mathfrak{E}_{\alpha}\mathfrak{U}_{\gamma}[f]=\frac1{|\alpha|}\mathfrak{U}_{\overline{\gamma}}[f], \quad \overline{\gamma}=2+\frac{\gamma-2}{\alpha},
\end{equation}
while for $\gamma=2$ we find
\begin{equation}\label{eq:compEU2}
\mathfrak{E}_{\alpha}[\mathfrak{U}_2[f]](y)=e^{-\alpha x(y)}, \quad y'(x)=-e^{\alpha x}f(x).
\end{equation}
\end{proposition}
\begin{proof}
Recalling from the definition of the differential-escort transformation that
$$
\frac{d}{dy}\mathfrak{E}_{\alpha}[f](y)=\alpha f(x)^{2\alpha-2}f'(x),
$$
we deduce that
\begin{equation*}
\begin{split}
\mathfrak{D}_{\gamma}\mathfrak{E}_{\alpha}[f]&=\frac{f_{\alpha}(y)^{\gamma}}{|f'_{\alpha}(y)|}=\frac{f(x(y))^{\gamma\alpha}}{|\alpha|f(x(y))^{2\alpha-2}|f'(x(y))|}\\
&=\frac{1}{|\alpha|}\frac{f(x)^{\alpha\gamma-2\alpha+2}}{|f'(x)|}=\frac{1}{|\alpha|}\mathfrak{D}_{2+\alpha(\gamma-2)}[f],
\end{split}
\end{equation*}
establishing thus \eqref{eq:compDE}.

The proof of \eqref{eq:compEU} follows quickly by observing that Proposition \ref{prop:inv} gives $\mathfrak D_\gamma \mathfrak E_{\frac 1\alpha}\mathfrak E_{\alpha}\mathfrak U\gamma=\mathbb I$. By applying \eqref{eq:compDE}, we infer that
$$
\mathfrak E_{\alpha}\mathfrak U_\gamma=\left(\mathfrak D_\gamma \mathfrak E_{\frac 1\alpha}\right)^{-1}=\left(|\alpha|\mathfrak D_{2+\frac{\gamma-2}\alpha}\right)^{-1}=\frac 1{|\alpha|}\mathfrak U_{2+\frac{\gamma-2}\alpha}\ ,
$$
as stated.

Let now $\gamma=2$ and notice that the definition of the up transformation gives
$$
\mathfrak{E}_{\alpha}[\mathfrak{U}_2[f]](y)=\mathfrak{U}_2[f](u(y))^{\alpha}=e^{-\alpha x(u(y))},
$$
where, again by the chain rule and identifying $x(u(y))\equiv x(y)$
$$
y'(x)=y'(u)u'(x)=e^{-(1-\alpha)x(u)}(-e^xf(x))=-e^{\alpha x(y)}f(x),
$$
ending the proof of \eqref{eq:compEU2}.
\end{proof}
The outcome of Proposition \ref{prop:comp_UDE} indicates that the conjugations of up/escort/down transformations of the form $\mathfrak{D}_{\gamma}\mathfrak{E}_{\alpha}\mathfrak{U}_{\gamma}$ appearing in the monotonicity properties in Propositions \ref{prop:monot_moment} and \ref{prop:monot_Fisher} are reduced in fact to compositions of down and up transformations with different parameters. Motivated by this discussion, and recalling from Proposition \ref{prop:inv} that a composition of the down and up transformations with the same parameter is the identity, we explore in a systematic form in the forthcoming lines the outcome of the composition of up and down transformations with different parameters. The first and most interesting result relates the up and down transformations with parameters different from two.
\begin{proposition}\label{prop:comp_UD}
Let $\alpha$, $\beta\in\Rset\setminus\{2\}$. Then we have

(a) For any probability density function $f$,
\begin{equation}\label{eq:compDU}
\mathfrak{D}_{\alpha}\mathfrak{U}_{\beta}[f]=\mathfrak{G}_{2-\alpha,2-\beta}[f],
\end{equation}
where, for any $p$, $q\neq0$,
$$
\mathfrak{G}_{p,q}[f](s):=|ps|^{\frac{q}{p}-1}f\left(\frac{1}{|q|}|ps|^{\frac{q}{p}}\right).
$$
(b) For any decreasing probability density function $f$,
\begin{equation}\label{eq:compUD}
\mathfrak{U}_{\beta}\mathfrak{D}_{\alpha}[f]=C(\alpha,\beta)\mathfrak{E}_{\xi(\alpha,\beta)}[f], \quad \xi(\alpha,\beta)=\frac{\beta(\alpha-1)-3\alpha+4}{(2-\beta)^2},
\end{equation}
where
$$
C(\alpha,\beta):=\left|\left|\frac{2-\beta}{2-\alpha}\right|^{\frac{1}{2-\beta}}\frac{1}{\xi(\alpha,\beta)}\right|^{\frac{1}{2-\beta}}.
$$
\end{proposition}
\begin{proof}
(a) Observing that
$$
|\up{\beta}[f]'(u)|=|(\beta-2)x(u)|^{\frac{\beta-1}{2-\beta}}|x'(u)|=|(\beta-2)x(u)|^{\frac{\beta}{2-\beta}}\frac{1}{f(x(u))},
$$
it follows from the definition of the down and up transformations that
\begin{equation}\label{eq:interm8}
\begin{split}
\mathfrak{D}_{\alpha}\mathfrak{U}_{\beta}[f](s)&=\frac{\up{\beta}[f]^{\alpha}(u(s))}{|\up{\beta}[f]'(u(s))|}
=|(\beta-2)x(u(s))|^{\frac{\alpha}{2-\beta}}\frac{f(x(u(s)))}{|(\beta-2)x(u(s))|^{\frac{\beta}{2-\beta}}}\\
&=|(\beta-2)x(u(s))|^{\frac{\alpha-\beta}{2-\beta}}f(x(u(s))),
\end{split}
\end{equation}
and it remains to express the double change of variable $x(u(s))\equiv x(s)$. The chain rule then gives
\begin{equation}\label{eq:interm12}
\begin{split}
\frac{ds}{dx}&=s'(u(x))u'(x)=(\up{\beta})[f]^{1-\alpha}(u(x))|(\up{\beta})[f]'(u(x))|(-|(\beta-2)x(u)|^{\frac{1}{\beta-2}})f(x(u))\\
&=-|(\beta-2)x(u)|^{\frac{1-\alpha}{2-\beta}}\frac{|(\beta-2)x(u)|^{\frac{\beta}{2-\beta}}}{f(x(u))}|(\beta-2)x(u)|^{\frac{1}{\beta-2}}f(x(u))\\
&=-|(\beta-2)x|^{\frac{\beta-\alpha}{2-\beta}}.
\end{split}
\end{equation}
We may assume without loss of generality that the support of $f$ is a subset of $(0,\infty)$, thus $x>0$ in the previous calculations. We obtain by integration that
$$
s(x)=-|\beta-2|^{\frac{\beta-\alpha}{2-\beta}}\frac{2-\beta}{2-\alpha}x^{\frac{2-\alpha}{2-\beta}},
$$
or, equivalently,
\begin{equation}\label{eq:interm9}
x(s)=\left[\frac{2-\alpha}{2-\beta}|\beta-2|^{\frac{\alpha-\beta}{2-\beta}}|s|\right]^{\frac{2-\beta}{2-\alpha}}
=\left(\frac{2-\alpha}{2-\beta}\right)^{\frac{2-\beta}{2-\alpha}}|\beta-2|^{\frac{\alpha-\beta}{2-\alpha}}|s|^{\frac{2-\beta}{2-\alpha}}.
\end{equation}
Inserting $x(s)$ from \eqref{eq:interm9} into \eqref{eq:interm8}, we conclude that
\begin{equation*}
\begin{split}
\mathfrak{D}_{\alpha}\mathfrak{U}_{\beta}[f](s)&
=\left||\beta-2|^{\frac{2-\beta}{2-\alpha}}\left|\frac{2-\alpha}{2-\beta}\right|^{\frac{2-\beta}{2-\alpha}}s^{\frac{2-\beta}{2-\alpha}}\right|^{\frac{\alpha-\beta}{2-\beta}}f(x(s))\\
&=|(2-\alpha)s|^{\frac{\alpha-\beta}{2-\alpha}}f\left(\frac{1}{|\beta-2|}|(2-\alpha)s|^{\frac{2-\beta}{2-\alpha}}\right),
\end{split}
\end{equation*}
which is obviously equivalent to \eqref{eq:compDU}.

\medskip

(b) By the definition of the up transformation,
\begin{equation}\label{eq:interm10}
\mathfrak{U}_{\beta}\mathfrak{D}_{\alpha}[f](u)=|(\beta-2)s(u)|^{\frac{1}{2-\beta}},
\end{equation}
where, again employing the changes of variable in the definitions of both the up and down transformations,
\begin{equation*}
\begin{split}
u'(s)&=-|(\beta-2)s|^{\frac{1}{\beta-2}}\adown(s)=-|(\beta-2)s(x)|^{\frac{1}{\beta-2}}\frac{f(x(s))^{\alpha}}{|f'(x(s))|}\\
&=-\left|\frac{\beta-2}{\alpha-2}f(x(s))^{2-\alpha}\right|^{\frac{1}{\beta-2}}\frac{f(x(s))^{\alpha}}{|f'(x(s))|}\\
&=-\left|\frac{\beta-2}{\alpha-2}\right|^{\frac{1}{\beta-2}}\frac{f(x(s))^{\frac{2-\alpha}{\beta-2}+\alpha}}{|f'(x(s))|}.
\end{split}
\end{equation*}
We can thus invert the previous equality to find (recalling that $x(s)\equiv x(s(u))$)
$$
s'(u)=-\left|\frac{\beta-2}{\alpha-2}\right|^{\frac{1}{2-\beta}}f(u)^{\frac{2-\alpha}{2-\beta}-\alpha}|f'(u)|
$$
and, since $f$ is assumed to be decreasing, we integrate to obtain
\begin{equation}\label{eq:interm11}
s(u)=\left|\frac{\beta-2}{\alpha-2}\right|^{\frac{1}{2-\beta}}\frac{2-\beta}{\alpha\beta-3\alpha-\beta+4}f(u)^{\frac{\alpha\beta-3\alpha-\beta+4}{2-\beta}}
\end{equation}
Inserting now $s(u)$ from \eqref{eq:interm11} into \eqref{eq:interm10}, we obtain
\begin{equation*}
\begin{split}
\mathfrak{U}_{\beta}\mathfrak{D}_{\alpha}[f](u)&
=\left|(\beta-2)\left|\frac{\beta-2}{\alpha-2}\right|^{\frac{1}{2-\beta}}\frac{2-\beta}{\alpha\beta-3\alpha-\beta+4}\right|^{\frac{1}{2-\beta}}
f(u)^{\frac{\alpha\beta-3\alpha-\beta+4}{(2-\beta)^2}}\\
&=C(\alpha,\beta)\mathfrak{E}_{\xi(\alpha,\beta)}[f],
\end{split}
\end{equation*}
completing the proof.
\end{proof}
In a similar way, one can deduce the expression of an application of combined up and down transformations when either $\alpha=2$ or $\beta=2$. We give these expressions below for the sake of completeness, despite the fact that their final expressions cannot be expressed in terms of other transformations, as it happened in the identities established in Proposition \ref{prop:comp_UD}.
\begin{proposition}\label{prop:comp_UD2}
Let $\alpha$, $\beta\in\Rset\setminus\{2\}$. Then we have

(a) For any probability density function $f$,
\begin{equation}\label{eq:comp_D2U}
\down{2}\up{\beta}[f]=\widetilde{\mathfrak{G}}_{2-\beta}[f], \quad \widetilde{\mathfrak{G}}_p[f](x):=|pe^{px}|f(e^{px}).
\end{equation}

(b) For any decreasing probability density function $f$,
\begin{equation}\label{eq:comp_UD2}
\up{\beta}\down{2}[f](u)=|(\beta-2)x(u)|^{\frac{1}{2-\beta}}, \quad u(x):=\int_x^{x_f}|(\beta-2)\log\,f(t)|^{\frac{1}{\beta-2}}f(t)\,dt.
\end{equation}

(c) For any probability density function $f$,
\begin{equation}\label{eq:comp_DU2}
\down{\alpha}\up{2}[f](s)=\frac{1}{(2-\alpha)s}f\left(\frac{\log((2-\alpha)s)}{\alpha-2}\right).
\end{equation}

(d) For any decreasing probability density function $f$,
\begin{equation}\label{eq:comp_U2D}
\up{2}\down{\alpha}[f](u)=e^{-x(u)}, \quad u(x):=\int_x^{x_f}\exp\left(\frac{f(t)^{2-\alpha}}{\alpha-2}\right)f(t)\,dt.
\end{equation}
\end{proposition}
\begin{proof}
(a) On the one hand, we obtain by simply letting $\alpha=2$ in \eqref{eq:interm8} that
\begin{equation}\label{eq:interm13}
\down{2}\up{\beta}[f](s)=|(\beta-2)x(u(s))|f(x(u(s)).
\end{equation}
On the other hand, \eqref{eq:interm12} reads for $\alpha=2$
$$
s'(x)=-\frac{1}{|(\beta-2)x|},
$$
and with the same convention that $x>0$, we find
$$
s(x)=\frac{1}{2-\beta}\log\,x, \quad {\rm or, \ equivalently,} \quad x(s)=e^{(2-\beta)s}.
$$
Inserting this expression of $x(s)\equiv x(u(s))$ into \eqref{eq:interm13} easily leads to \eqref{eq:comp_D2U}.

\medskip

(b) We infer from the definition of the up transformation that
\begin{equation}\label{eq:interm14}
\up{\beta}\down{2}[f](u)=|(\beta-2)s(u)|^{\frac{1}{2-\beta}},
\end{equation}
with
$$
u'(s)=-|(\beta-2)s|^\frac1{\beta-2}f^{\downarrow}_{2}(s)=-|(\beta-2)\log\,f(x)|^{\frac{1}{\beta-2}}\frac{f(x(s))^2}{|f'(x(s))|}.
$$
In addition, from the change of variable inside the definition of the down transformation,
$$
s'(x)=\frac{|f'(x)|}{f(x)}
$$
from where, applying the chain rule
$$
\frac {du}{dx}=u'(s(x))s'(x)=-|(\beta-2)\log\,f(x)|^{\frac{1}{\beta-2}}f(x(s)),
$$
which easily leads to \eqref{eq:comp_UD2} after an integration and inserting its result into \eqref{eq:interm14}.

\medskip

(c) We first deduce from the definitions of both the up and down transformations that
\begin{equation}\label{eq:interm15}
\down{\alpha}\up{2}[f](s)=\frac{\up{2}[f]^{\alpha}(u(s))}{|\up{2}'[f](u(s))|}=e^{(2-\alpha)x(u(s))}f(x(u(s))).
\end{equation}
Observing that the definition of the up transformation with parameter $\beta=2$ gives
$$
\frac{d}{du}\up{2}[f](u)=-e^{-x(u)}x'(u)=\frac{e^{-2x(u)}}{f(x(u))},
$$
we derive by an application of the chain rule that
\begin{equation*}
\begin{split}
s'(x)&=s'(u(x))u'(x)=-\left(\up{2}[f]^{1-\alpha}\left|\frac{d}{du}\up{2}[f]\right|\right)(u(x))e^{x}f(x)\\
&=-e^{-(1-\alpha)x}\frac{e^{-2x}}{f(x)}e^{x}f(x)=-e^{(\alpha-2)x},
\end{split}
\end{equation*}
and it follows by an integration and inversion that
\begin{equation}\label{eq:interm16}
s(x)=\frac{1}{2-\alpha}e^{(\alpha-2)x}, \quad {\rm that \ is}, \quad x(s)=\frac{1}{\alpha-2}\log((2-\alpha)s).
\end{equation}
Inserting the outcome of \eqref{eq:interm16} into \eqref{eq:interm15} leads immediately to \eqref{eq:comp_DU2}.

\medskip

(d) At first, we have
\begin{equation}\label{eq:interm17}
\up{2}\down{\alpha}[f](u)=e^{-s(u)},
\end{equation}
where, taking into account the monotonicity of $f$ and the canonical election in the down transformation,
\begin{equation*}
u'(s)=-e^{s}\adown(s)=e^{s(x)}\frac{f(x(s))^{\alpha}}{f'(x(s))}=\exp\left(-\frac{f(x(s))^{2-\alpha}}{2-\alpha}\right)\frac{f(x(s))^{\alpha}}{f'(x(s))}.
\end{equation*}
In addition, by applying the chain rule one has
$$
\frac{du}{dx}=u'(s(x))s'(x)=\exp\left(-\frac{f(x)^{2-\alpha}}{2-\alpha}\right)\frac{f(x)^{\alpha}}{f'(x)}|f'(x)|f(x)^{1-\alpha}=-\exp\left(-\frac{f(x)^{2-\alpha}}{2-\alpha}\right)f(x)
$$
Then \eqref{eq:comp_U2D} follows from \eqref{eq:interm17}.
\end{proof}

\subsection*{Acknowledgements}

R. G. I. is partially supported by the project PID2024-160967NB-100 (AEI) funded by the Ministry of Science, Innovation and Universities of Spain. D. P.-C. is partially supported by the project PID2023-153035NB-100 (AEI) funded by the Ministry of Science, Innovation and Universities of Spain and “ERDF/EU A way of making Europe”.

\bigskip

\noindent \textbf{Data availability} Our manuscript has no associated data.

\bigskip

\noindent \textbf{Competing interest} The authors declare that there is no competing interest.

\bibliographystyle{unsrt}
\bibliography{refs}
\end{document}